\documentclass[a4paper,onecolumn,11pt,accepted=2023-02-02]{quantumarticle}
\pdfoutput=1

\usepackage{amsmath}
\usepackage{amsfonts}
\usepackage{amssymb}
\usepackage{amsthm}
\usepackage{authblk}
\usepackage{graphicx}
\usepackage{hyperref}
\usepackage[square,semicolon,numbers,sort]{natbib}

\usepackage[utf8]{inputenc}
\usepackage[english]{babel}
\usepackage[T1]{fontenc}
\usepackage{amsmath}

\usepackage{tikz}
\usepackage{lipsum}
\usepackage{fullpage}

\usepackage{enumerate}

\usepackage{adjustbox}
\usepackage{soul}
\usepackage{comment}
\usepackage{bbold}
\usepackage{tikz}
\usetikzlibrary{decorations.pathreplacing,angles,quotes}
\usepackage{bbold}

\usepackage{diagbox}

\usetikzlibrary{matrix,arrows,decorations.pathmorphing}
\usepackage{tikz-cd}
\usetikzlibrary{arrows} 

\usetikzlibrary{decorations.markings}
 \usepackage{caption}
\usepackage{subcaption}
 
 \usepackage{pdfpages}

\usepackage{centernot}
\usepackage{mathtools}
\usepackage{stmaryrd}

\newtheorem{theorem}{Theorem}

\newtheorem{corollary}{Corollary}
\newtheorem{lemma}{Lemma}
\newtheorem{definition}{Definition}
\newtheorem*{question*}{Question}
\newtheorem{example}{Example}

\newtheorem{remk}{Remark}

\newcommand{\rev}[1]{#1}

\newcommand{\ra}{\ensuremath{\rightarrow}}
\newcommand{\zz}{\mathbb Z}

\newcommand{\R}{\mathbb R}
\newcommand{\CC}{\mathbb C}
\newcommand{\Z}{\mathbb Z}
\newcommand{\N}{\mathbb N}
\newcommand\mc[1]{\mathcal{#1}}
\newcommand\mH{\mathsf{H}}

\newcommand{\suchthat}{\;\ifnum\currentgrouptype=16 \middle\fi|\;}

\newcommand{\one}{\mathbbm{1}}
\newcommand{\set}[1]{\ensuremath{ \lbrace #1 \rbrace }}
\newcommand{\Span}[1]{\ensuremath{ \langle #1 \rangle }}

  \title{No quantum solutions to linear constraint systems in odd dimension from Pauli group and diagonal Cliffords}
\author{Markus Frembs}
\email{m.frembs@griffith.edu.au}
\affiliation{Centre for Quantum Dynamics, Griffith University,\\ Yugambeh Country, Gold Coast, QLD 4222, Australia}
\author{Cihan Okay}
\email{cihan.okay@bilkent.edu.tr}
\affiliation{Department of Mathematics, Bilkent University, Ankara, Turkey}
\orcid{0000-0001-8097-5227}
\author{Ho Yiu Chung}
\email{hoyiu.chung@bilkent.edu.tr}
\affiliation{Department of Mathematics, Bilkent University, Ankara, Turkey}
  
\begin{document}
  \maketitle  
  
\vspace{-.3cm}
\begin{abstract}
    Linear constraint systems (LCS) have proven to be a surprisingly prolific tool in the study of non-classical correlations and various related issues in quantum foundations. Many results are known for the Boolean case, yet the generalisation to systems of odd dimension is largely open. In particular, it is not known whether there exist LCS in odd dimension, which admit finite-dimensional quantum, but no classical solutions. In recent work,  [J. Phys. A,  \textbf{53},  385304 (2020)] have shown that unlike in the Boolean case, where the $n$-qubit Pauli group gives rise to quantum solutions of LCS such as the Mermin-Peres square, the $n$-qudit Pauli group never gives rise to quantum solutions of a LCS in odd dimension. Here, we generalise this result towards the Clifford hierarchy. More precisely, we consider tensor products of groups generated by (single-qudit) Pauli and diagonal Clifford operators. 
\end{abstract}
\vspace{-1.1cm}

\tableofcontents

\section{Introduction}

Linear constraint systems (LCS) have sparked a lot of interest in recent years, see \cite{Arkhipov2012,CleveMittal2014,Fritz2016,CleveLiuSlofstra2017,Slofstra2019,OkayRaussendorf2020} for instance. Every LCS can be cast into a two-player nonlocal game, which from a computational perspective are of interest as multi-prover interactive proof systems \cite{CleveMittal2014}. In this form, LCS have played a prominent role in the various stages and final solution of Tsirelson's problem \cite{Tsirelson2006,Slofstra2019,Slofstra2019b,JiEtAl2020}. The surprisingly rich interplay between quantum foundations and computer science, as well as the structure theory of von Neumann algebras \cite{Kirchberg1993,Fritz2012,JungeEtAl2011} shows that LCS have important applications in quantum theory.

Quantum solutions to LCS are known to exist only in even or infinite (Hilbert space) dimension;\footnote{Here, the dimension refers to the (minimal) Hilbert space on which operators act. The existence of linear constraint systems over $\zz_d$ admitting possibly infinite-dimensional quantum but no classical solutions for arbitrary $d \in \mathbb{N}$ has been reported in \cite{ZhangSlofstra2020}.}\label{fn: ZhangSlofstra} it is an open question whether finite-dimensional quantum solutions to LCS also exist over $\zz_d$ for $d$ odd. Notably, the qubit Pauli operators in the Mermin-Peres square (see Fig.~\ref{fig: MP square} below) define a quantum solution to a LCS over $\zz_2$ (see Eq. ~(\ref{eq: MP-square LCS})). This raises the question whether LCS with quantum but no classical solutions exist within the $n$-qudit Pauli group $\mc{P}^{\otimes n}_d$ for arbitrary finite dimension $d$; yet, \cite{QassimWallman2020} have shown that this is not the case.

\rev{Here, we generalise this result to a larger group, taking inspiration from (measurement-based) quantum computation, where the Pauli group appears at the first level of the so-called Clifford hierarchy (see Sec. ~\ref{sec: Diagonal Clifford hierarchy} below) \cite{Gottesman1998,GottesmanChuang1999,deBeaudrap2013}.} More precisely, we consider extensions of the (single-qudit) Pauli group by diagonal Clifford operators. Importantly, the resulting group contains a set of operators that has been shown to be universal for function computation in the restricted subtheory of deterministic, non-adaptive \rev{measurement-based quantum computation (MBQC) with linear side-processing (on a GHZ resource state)} (see Ref.~\cite{Raussendorf2013,FrembsRobertsBartlett2018,FrembsRobertsCampbellBartlett2023} for details). Within this subtheory, contextuality---roughly, the impossibility of assigning outcomes to all measurements independent of other simultaneously performed measurements---has been identified as a resource for computational advantage \cite{Raussendorf2013,FrembsRobertsBartlett2018}. Indeed,  Pauli and diagonal Clifford measurement operators generally lead to contextual MBQC \cite{FrembsRobertsCampbellBartlett2023}. Since the nonexistence of classical solutions to a LCS is also a form of contextuality, this further motivates to seek quantum solutions to LCS for $\zz_d$ within this group.\footnote{MBQC and LCS are related, in fact, every deterministic, non-adaptive MBQC acting on a GHZ state with operators from this enlarged group gives rise to an associated LCS.  Note that measurement operators in MBQC commute on their common resource eigenstate. Consequently, the MBQC lifts to a quantum solution of its associated LCS if and only if measurement operators commute (independent of the resource state). Since it would distract from the results of the present paper, we will defer a more detailed analysis of the correspondence between MBQC and LCS to future work.} Surprisingly, our main result (Thm.~\ref{thm: main result - Clifford hierarchy}) shows that no such solutions exist. \rev{What is more, the techniques developed in this work constitute a first step towards the analysis of (the existence of quantum solutions to) LCS in groups generated from other gate sets used in quantum computation (beyond Pauli and diagonal Clifford gates).}

The rest of this paper is organised as follows. In Sec.~\ref{sec: LCS and MP square} we briefly review the Mermin-Peres square and the definition of linear constraint systems. In Sec.~\ref{sec: MBQC-group}, we  define the group of interest, \rev{generated from Pauli and diagonal Clifford gates.} Sec.~\ref{sec: group properties} studies its properties, and builds up to the proof of our key technical contribution (Thm.~\ref{thm: quasi-local homomorphism in abelian subgroups of order p}). The details are of independent interest, but  may be skipped on first reading. Sec.~\ref{sec: from noncommutative to commutative LCS} contains the main result of our work (Thm.~\ref{thm: main result - Clifford hierarchy}), and Sec.~\ref{sec: discussion} concludes.

\section{Linear constraint systems and Mermin-Peres square}\label{sec: LCS and MP square}

\begin{figure}[!htb]
    \centering
        \begin{tikzpicture}[x=0.08mm,y=0.08mm]
            
            \node at (0,900)   (A) {\rev{$x_1 = X \otimes \one$}};
            \node at (300,900)   (B) {\rev{$x_2 = \one \otimes Y$}};
            \node at (600,900)   (C) {\rev{$x_3 = X \otimes Y$}};
            \node at (0,600)   (D) {\rev{$x_4 = \one \otimes X$}};
            \node at (300,600)   (E) {\rev{$x_5 = Y \otimes \one$}};
            \node at (600,600)   (F) {\rev{$x_6 = Y \otimes X$}};
            \node at (0,300)   (G) {\rev{$x_7 = X \otimes X$}};
            \node at (300,300)   (H) {\rev{$x_8 = Y \otimes Y$}};
            \node at (600,300)   (J) {\rev{$x_9 = Z \otimes Z$}};
            
            \draw[thick] (A) -- (B);
            \draw[thick] (B) -- (C);
            \draw[thick] (A) -- (D);
            \draw[thick] (B) -- (E);
            \draw[thick] (C) -- (F);
            \draw[thick] (D) -- (E);
            \draw[thick] (E) -- (F);
            \draw[thick] (D) -- (G);
            \draw[thick] (E) -- (H);
            \draw[thick] (F) -- (J);
            \draw[thick] (G) -- (H);
            \draw[thick] (H) -- (J);
        \end{tikzpicture}
    \caption{\rev{Two-qubit Pauli operators in the (contextuality proof of the) Mermin-Peres square \cite{Mermin1990,Mermin1993}.}}
    \label{fig: MP square}
\end{figure}

\rev{Let $\mc{P}^{\otimes n}_2 = \langle X,Y,Z\rangle^{\otimes n}$ denote the $n$-qubit Pauli group.\footnote{\rev{The Pauli matrices are $X = \begin{pmatrix} 0 & 1 \\ 1 & 0 \end{pmatrix}$, $Y = \begin{pmatrix} 0 & -i \\ i & 0 \end{pmatrix}$ and $Z = \begin{pmatrix} 1 & 0 \\ 0 & -1 \end{pmatrix}$.}}} The nine two-qubit Pauli operators $x_k \in \mc{P}^{\otimes 2}_2$, $k\in\{1,\cdots,9\}$ in Fig.~\ref{fig: MP square} satisfy the following constraints: \rev{(i) every operator is a square root of the identity, $x_k^2 = \one \otimes \one$ for all $k\in\{1,\cdots,9\}$, (ii) the operators in every row and column commute, $[x_k,x_l] = 0$ for all $k,l\in\{1,\cdots,9\}$}, and (iii) the operators in every row and column obey the multiplicative constraints
\begin{equation}\label{eq: MP-square constraints}
    \prod_{k=1}^9 x_k^{A_{ik}} = (-1)^{b_i}\; , \quad \rev{\forall i \in \{1,\cdots,6\}}\; ,
\end{equation}
where $A \in M_{6,9}(\zz_2)$ and $b \in \zz^6_2$ are defined as follows
\begin{equation}\label{eq: MP-square LCS}
    A = \begin{pmatrix}
         1 & 1 & 1 & 0 & 0 & 0 & 0 & 0 & 0  \\
         0 & 0 & 0 & 1 & 1 & 1 & 0 & 0 & 0  \\
         0 & 0 & 0 & 0 & 0 & 0 & 1 & 1 & 1  \\
         1 & 0 & 0 & 1 & 0 & 0 & 1 & 0 & 0  \\
         0 & 1 & 0 & 0 & 1 & 0 & 0 & 1 & 0  \\
         0 & 0 & 1 & 0 & 0 & 1 & 0 & 0 & 1  \\
    \end{pmatrix} \quad \quad \quad b = \begin{pmatrix}
         0 \\ 0 \\ 0 \\ 0 \\ 0 \\ 1
    \end{pmatrix} \; .
\end{equation}
Assume that there exists a (classical) solution $x_k = (-1)^{\tilde{x}_k}$ with $\tilde{x}_k \in \zz^9_2$. Then the constraints can be equivalently expressed in terms of the additive equation $A\tilde{x} = b \mod 2$,
\begin{equation}\label{eq: LCS relation}
    \prod_{k=1}^9 x_k^{A_{ik}} = (-1)^{\sum_{k=1}^9 A_{ik}\tilde{x}_k} = (-1)^{b_i}\; .
\end{equation}
However, as is easily seen by multiplying the constraints in Eq.~(\ref{eq: MP-square constraints}), no such solution exists. The Mermin-Peres square thus constitutes a proof of \emph{quantum contextuality}, that is, the impossibility of assigning spectral values to observables independent of other commuting, and thus simultaneously measurable observables.  \rev{More generally, we will call a group $G \subset U(d)$ \emph{noncontextual} if it admits a spectral value assignment, i.e., a map $v: G \ra U(1)$ such that (i) $v(g) \in \mathrm{sp}(g)$ and (ii) $v(gg') = v(g)v(g')$ whenever $[g,g'] := gg'g^{-1}g'^{-1} = \one$.\footnote{Here, $\mathrm{sp}(g) = \{\lambda \in U(1) \mid \det(g-\lambda\one)=0\}$ denotes the \emph{spectrum of $g$}, that is, the set of eigenvalues of $g$. } Otherwise $G$ is called \emph{contextual}. Note that any classical solution to Eq.~(\ref{eq: MP-square LCS}) would amount to a value assignment for $\mc{P}^{\otimes 2}_2$ in Fig.~\ref{fig: MP square}.}

Generalising this example, one defines linear constraint systems (LCS) `$Ax=b \mod d$' for any $A \in M_{m, n}(\zz_d)$ and $b \in \zz^m_d$.  A \emph{classical solution} to a LCS is a vector $\tilde{x} \in \zz_d^n$ solving the system of linear equations $A\tilde{x}=b \mod d$, whereas a \emph{quantum solution} is an assignment of unitary operators $x \in U^{\otimes n}(d)$ such that (i) $x_k^d = \mathbbm{1}$ for all $k \in \{1,\cdots,n\}$ (`$d$-torsion'), (ii) $[x_k,x_{k'}] := x_kx_{k'}x_k^{-1}x_{k'}^{-1} = \mathbbm{1}$ whenever $(k,k')$ appear in the same row of $A$, i.e., $A_{lk} \neq 0 \neq A_{lk'}$ for some $l \in \{1,\cdots,m\}$ (`commutativity'), and (iii) $\prod_{k=1}^n x^{A_{lk}} = \omega^{b_l}$ for all $l \in \{1,\cdots,m\}$, and where $\omega = e^{\frac{2\pi i}{d}}$ `constraint satisfaction').\footnote{With \cite{OkayRaussendorf2020}, we call an \emph{operator solution} an assignment $x \in U^{\otimes n}(d)$ satisfying (i) and (iii), but not (ii).} For more details on linear constraint systems we refer to \cite{CleveLiuSlofstra2017,Slofstra2019,QassimWallman2020}.\\

\section{On Paulis and Cliffords}\label{sec: MBQC-group}

This section builds up to the definition of the groups $K_Q^{\otimes n}$ in Sec.~\ref{sec: definition K-group}, which encompasses Pauli and diagonal Clifford operators. We first review some basic facts about the Heisenberg-Weyl group in Sec.~\ref{sec: Pontryagin duality}, before we extend it by additional operators, motivated by the diagonal Clifford hierarchy as shown in Sec.~\ref{sec: Diagonal Clifford hierarchy}.

\subsection{Pontryagin duality and Heisenberg-Weyl group}\label{sec: Pontryagin duality}

Given any \rev{finite} abelian group $G$, we define its Pontryagin dual $\hat{G}$ as the set of group homomorphisms into the unitary group $U(1)$,
\begin{equation}\label{eq: Pontryagin dual}
    \hat{G} := \mathrm{Hom}(G,U(1)) = \{\chi: G \rightarrow U(1)\ \mid \chi(gh) = \chi(g)\chi(h)\ \forall g,h \in G\}\; .
\end{equation}
$\hat{G}$ itself is a \rev{finite} abelian group under pointwise multiplication.  By Pontryagin duality,  there is a canonical isomorphism $\mathrm{ev}_G: G \rightarrow \hat{\hat{G}}$ given by
    \begin{equation*}
        \mathrm{ev}_G(x)(\chi) := \chi(x)\; .
    \end{equation*}
The Heisenberg-Weyl group $\mH(G)$ of $G$ is now defined as follows. Let \rev{$\mathbb{C}^G$} be the Hilbert space of complex-valued functions on $G$ and define the translation operators $t_x$, $x \in G$ and multiplication operators $m_\chi$, $\chi \in \hat{G}$ by
\begin{equation}\label{eq: GW generators}
    (t_x f)(y) = f(x+y) \quad \quad \quad
    (m_\chi f)(y) = \chi(y)f(y)\; ,
\end{equation}
for all $f \in \rev{\mathbb{C}^G}$.\footnote{The measure needed in the definition of $L^2(G)$ is the unique Haar measure on $G$.} Clearly, the operators in Eq.~(\ref{eq: GW generators}) do not commute, instead they satisfy the canonical Weyl commutation relations,
\begin{equation}\label{eq: Weyl CCR}
    t_xm_\chi = \overline{\chi(x)} m_\chi t_x\; , \quad \forall x \in G, \chi \in \hat{G}\; .
\end{equation}
The Heisenberg-Weyl group is the subgroup of the unitary group on the Hilbert space \rev{$\mathbb{C}^G$} generated by these operators. 

Of particular interest for quantum computation is the discrete Heisenberg-Weyl group $\mH(\zz_d)$ with $\hat{\zz}_d = \{\chi^a(q) = \omega^{aq} \mid \omega = e^{\frac{2\pi i}{d}}, a \in \zz_d\} \cong \zz_d$.\footnote{For $d$ odd, $\mc{P}_d \cong \mH(\zz_d)$, but for $d$ even these groups differ by central elements \cite{deBeaudrap2013}. In particular, note that the qubit Pauli group $\mc{P}_2 = \langle i,X,Z\rangle$, whereas $\mH(\zz_2) = \langle X,Z \rangle$.} Note that the Heisenberg-Weyl group can also be seen as a subgroup of the normaliser $\mH(\zz_d) \subset N(T)$ of a maximal torus of the special unitary group,
\begin{align}\label{eq: maximal torus SU(d)}
    T= T(SU(d)) = \{S_\xi = \mathrm{diag}(\xi(0),\cdots,\xi(d-1)) \in SU(d) \mid \xi: \zz_d \ra U(1), \prod_{q\in \zz_d} \xi(q)=1\}\; .
\end{align}
\rev{See \cite[Chapter 11]{Hall} for details about maximal tori in Lie groups.} There is a split extension
\begin{equation}\label{eq: split extension - normaliser}
	1\to T\to N(T)\to W(T) \to 1\; ,
\end{equation}
where $N(T)$ is the normalizer of $T$ and $W(T) = N(T)/T$ is called the Weyl group.  \rev{In the case of the special unitary group $SU(d)$,  $W(T) \cong \mathrm{Sym}(d)$ is the symmetric group on $d$ elements.} It induces an action on $T$ by
\begin{equation}{\label{eq-action}}
x\cdot t=xtx^{-1}\; ,
\end{equation}
where $x\in W(T) \rev{\cong \mathrm{Sym}(d)}$ on $t\in T$.  Let $X$ 
be the generalised Pauli shift operator. The group $\langle X\rangle\cong\zz_d$ generated by $X$ can be identified with a subgroup of $W(T) \rev{\cong \mathrm{Sym}(d)}$.  In particular,  note that the Heisenberg-Weyl group arises from the split extension in Eq.~(\ref{eq: split extension - normaliser}), by restriction to the Pontryagin dual
\begin{equation}\label{eq: split extension - HW group}
	1\to \hat{\zz}_d \to \mH(\zz_d) \to \langle X\rangle \to 1\; .
\end{equation}
Comparing Eq.~(\ref{eq: split extension - HW group}) with Eq.~(\ref{eq: split extension - normaliser}) suggests to extend the Pontryagin dual $\hat{\zz}_d$ in the definition of the (discrete) Heisenberg-Weyl group to a maximal torus subgroup of the special unitary group.\footnote{\rev{The restriction to the special unitary group $SU(d)$ in Eq.~(\ref{eq: maximal torus SU(d)}) ensures that (local measurement) operators $M(\xi,b) = S_\xi X^b$,  with $b\neq 0$ and $S_\xi \in T$ have order $d$ (Lm.~\ref{lem-torsion-K}). We also remark that the extension of $\hat{\zz}_d$ to $T$ is closely related to the diagonal Clifford hierarchy (cf. \cite{CuiGottesmanKrishna2017}), which embeds into the maximal torus of the unitary group $T(U(d))$ (see also Remark.~\ref{rm: diagonal Clifford hierarchy}).}}

\subsection{Diagonal Clifford hierarchy}\label{sec: Diagonal Clifford hierarchy}

\textbf{Clifford hierarchy.} Recall that the \emph{Clifford hierarchy} is defined recursively from $\mc{C}_1(d) := \mc{P}_d$ and $\mc{C}_{k+1}(d) := \{C \in U(d) \mid \forall M \in \mc{P}_d: CMC^\dagger \subset \mc{C}_k(d)\}$, where we follow the notation in Ref.~\cite{deBeaudrap2013}.  $\mc{C}_2(d)$ is also known as the \emph{Clifford group}, it is the normaliser of the qudit Pauli group $\mc{P}_d$; note that $\mc{C}_k(d)$ is not a group for $k>2$. We remark that the Clifford hierarchy is of fundamental importance in quantum computation with magic state injection \cite{GottesmanChuang1999,HowardVala2012,BengtssonBlanchfieldCampbellHoward2014,Howard2015}. In particular, the Gottesman-Knill theorem proves that the so-called stabiliser subtheory \cite{Gottesman1997,Fujii2015}, which encompasses all unitary Clifford circuits on eigenstates of $n$-qubit Pauli group, and Pauli-measurements can be efficiently classically simulated \cite{Gottesman1998,AaronsonGottesman2004,deBeaudrap2013}. While stabiliser subtheory is not universal for quantum computation, adding a single additional unitary gate, e.g. the $T$-gate ($\frac{\pi}{4}$-phase gate) to any set of generators of the Clifford group does yield a universal gate set. 

\textbf{Diagonal Clifford hierarchy.} Similarly, the \emph{diagonal Clifford hierarchy} for $d$ odd is defined recursively from $\mc{D}_1(d) := \widehat{\zz}_d$ (see Sec. ~\ref{sec: Pontryagin duality}), and $\mc{D}_{k+1}(d) := \{S_\gamma = \mathrm{diag}(\gamma(0),\cdots,\gamma(d-1)) \in U(d) \mid \forall M \in \mH(\zz_d):\ S_\gamma MS_\gamma^\dagger \subset \mc{D}_k(d)\}$.\footnote{The definition of $(\mc{D}_k(d))_{k \in \mathbb{N}}$ notably depends on a choice of basis, viz. maximal torus $T(U(d)) \subset U(d)$ (cf. Eq.~(\ref{eq: maximal torus SU(d)}) below).} 
We will be particularly interested in the subgroups of the \emph{special diagonal Clifford hierarchy} $\mc{SD}_k(d) := \mc{D}_k(d) \cap SU(d)$ for all $k \in \mathbb{N}$, since operators $S_\xi X^b$ with $S_\xi \in \mc{SD}_k(d)$ are $d$-torsion by Lm.~\ref{lem-torsion-K} below. 

\begin{lemma}\label{lm: diagonal Clifford hierarchy}
	Let $\xi: \zz_d \ra U(1)$ be a function with $d$ odd. If $S_\xi \in \mc{SD}_k(d)$ then $S^{d^{k+1}}_\xi = \one$.
\end{lemma}

\begin{proof}
	The case $k=1$ is immediate. The general case $k>1$ follows by induction from the recursive definition of the (special) diagonal Clifford hierarchy. More precisely, for any diagonal element $S_\xi$ with $\xi: \zz_d \ra U(1)$ the canonical Weyl commutation relations in Eq.~(\ref{eq: Weyl CCR}) generalise to
	\begin{equation}\label{eq: d operator}
    	S_\xi X S^\dagger_\xi = S_{\Delta\xi} X\; ,
	\end{equation}
	where we define $\Delta\xi: \zz_d \rightarrow U(1)$ by $\Delta\xi(q) :=  \frac{\xi(q+1)}{\xi(q)} \in U(1)$. In particular,
	\begin{equation*}
    	\prod_{q=0}^{d-1} \Delta\xi(q)
    	= \prod_{q=0}^{d-1} \frac{\xi(q+1)}{\xi(q)}
    	= \frac{\prod_{q=0}^{d-1}\xi(q+1)}{\prod_{q=0}^{d-1}\xi(q)}
    	= \frac{1}{1} = 1\; .
	\end{equation*}
	By definition, $\mc{SD}_{k+1}(d)$ is the pre-image to $\mc{SD}_k(d)$ of the operator $\Delta: T \ra T$ for $T$ a maximal torus in $SU(d)$. 	Now, assume that the assertion holds for $k-1$ and let $S_{\xi'} \in \mc{SD}_{k-1}(d)$. From $\xi'(q) = \Delta\xi(q) = \frac{\xi(q+1)}{\xi(q)}$ we find that any pre-image $S_\xi \in \mc{SD}_k(d)$ of $S_{\xi'}$ under $\Delta$ satisfies $\xi(q+1) = \xi(q)\xi'(q)$. Let $\xi(0) = e^{ic}$ with $c \in \R$ arbitrary, then $\xi(q) = e^{ic}\prod_{k=0}^{q-1}\xi'(k)$. Moreover, since $S_\xi \in SU(d)$ by assumption we have $\xi^{-1}(0) = \prod_{q=1}^{d-1}\xi(q) = e^{i(d-1)c}\prod_{q=1}^{d-1}\prod^{q-1}_{k=0}\xi'(k)$, from which it follows that $e^{idc} = \prod_{q=1}^{d-1}\prod^{q-1}_{k=0}\xi'(k)$. Finally, since $\xi'(k)$ is a $k$-th root of unity by the inductive hypothesis, it follows that $e^{ic}$ is a $k+1$-th root of unity.
\end{proof}

It follows that not every operator $S_\xi$ with $\xi: \zz_d \ra U(1)$ is an element of $\mc{D}_k(d)$ for some $k\in\N$. For instance, $\xi = (1,e^{i\sqrt{2}},e^{-i\sqrt{2}}) \notin \mc{D}_k(3)$ for any $k \in \mathbb{N}$. Nevertheless, as a consequence of Thm.~\ref{thm-poly} in App.~\ref{sec: app K_Q(p)}, $S_\xi \in \mc{D}_k(d)$ if and only if $S^m_\xi = \one$ for some $m \in \mathbb{N}$. 

\rev{Contextuality has been identified as a resource in measurement-based quantum computation (MBQC) \cite{Raussendorf2013,FrembsRobertsBartlett2018}. Motivated in part by the universality of Pauli and diagonal Clifford gates within the restricted subtheory of deterministic, non-adaptive MBQC with linear side-processing (on a GHZ resource state) \cite{FrembsRobertsCampbellBartlett2023}, in the next section, we will define groups generated from tensor products of these operators.}

\subsection{Definition of $K_Q^{\otimes n}(d)$}\label{sec: definition K-group}

In this section, we define a group generated from operators $S_\xi X^b$ that are $d$-torsion in the sense that  $(S_\xi X^b)^d=\one$. For an integer $m\geq 1$, let $T_{(d^m)}$ denote the subgroup of the maximal torus of the special unitary group $T$ consisting of elements of order $d^m$; that is $T_{(d^m)} =\set{S_\xi\in T|\, S_\xi^{d^m}=\one}$. In particular, when $d$ is an odd prime, the subgroup $T_{(d)}$ consists of elements $S_\xi$ where $\xi$ is of the form  $\xi(q)=a_0+a_1q+a_2q^2+\cdots+a_{d-1}q^{d-1}$ with $a_i\in \zz_d$ for all $0\leq i\leq d-1$ (cf. Eq.~(\ref{eq:xi})).

\begin{definition}\label{def: local MBQC-group}
For a subgroup $T_{(d)}\subset Q\subset T_{(d^m)}$  invariant under the action of $\Span{X}$ we define the following subgroup of the special unitary group:     
\[K_Q(d) = \langle S_\xi X^b\suchthat  b \in \zz_d,S_\xi \in Q \rangle \subset SU(d)\; ,\]
 where $X|q\rangle = |q+1\rangle$ is the generalised Pauli shift operator and $S_\xi|q\rangle = \xi(q)|q\rangle$ is a generalised phase gate. \footnote{Note that not all (diagonal) elements in $K_Q(d)$ have order $d$. Such operators do not arise as local measurement operators in (deterministic, non-adaptive) $ld$-MBQC, they are merely a byproduct of Def.~\ref{def: local MBQC-group}.}

We denote the $n$-fold Kronecker tensor product of $K_Q(d)$ by
     \[K^{\otimes n}_Q(d)=\overbrace{K_Q(d)\otimes\cdots\otimes K_Q(d)}^{\text{n copies}}.\]
Similarly $H^{\otimes n}(\zz_d)$ denotes the 
    $n$-fold Kronecker tensor product of $H(\zz_d)$.
\end{definition}
 
Note that $K_Q(d) = \mH(\zz_d)$ reduces to the discrete Heisenberg-Weyl group\footnote{For $d$ odd, $\mc{P}_d \cong \mH(\zz_d)$, but for $d$ even these groups differ by central elements \cite{deBeaudrap2013}. In particular, note that the qubit Pauli group $\mc{P}_2 = \langle i,X,Z\rangle$, whereas $\mH(\zz_2) = \langle X,Z \rangle$.} for $Q=T_{(d)}$ where $d$ is an odd prime. In this paper we will consider the other extreme that is when $Q=T_{(d^m)}$.  

Note also that $K_Q(d) = Q \rtimes \langle X \rangle \cong Q \rtimes \zz_d$ is a split extension,  (as a restriction of Eq.~(\ref{eq: split extension - normaliser})) 
\begin{equation}\label{eq: split extension}
    1 \ra Q \ra K_Q(d) \ra \langle X \rangle \ra 1\; .
\end{equation}
We will therefore often write $S_\xi X^b\in K_Q(d)$ as $(\xi,b)$, where $b\in\Z_d$ and $S_\xi\in Q$. By Eq. (\ref{eq-action}), the group action of $\Z_d$ on $Q$ is given by
\[S_{b\cdot\xi}|q\rangle:=X^b\cdot S_\xi|q\rangle=X^bS_\xi X^{-b}|q\rangle=X^bS_\xi |q-b\rangle=X^b\xi(q-b)|q-b\rangle=\xi(q-b)|q\rangle\; ,\]
that is $b\cdot\xi(q)=\xi(q-b)$ for all $q\in\zz_p$.  The group operation of $K_Q(d)$ is given by
\[(\xi,b)(\xi',b')=(\xi b\cdot \xi',b+b')\; ,\]
and the inverse is given by
\[(\xi,b)^{-1}=((-b)\cdot \xi^{-1},-b)\; ,\]
where $(\xi,b),(\xi',b')\in K_Q(d)$.

\begin{lemma}{\label{lem-for}}
Let $(\xi,b),(\chi,0)\in K_Q(d)$. Let $n$ be a positive integer, we have
\begin{enumerate}[(1)]

\item $(\xi,b)^n=\left(\prod_{i=0}^{n-1}(ib)\cdot\xi \, , \, nb\right)$.

\item $(\xi,b)^n(\chi,0)=((nb)\cdot\chi,0)(\xi,b)^n$. Thus we also have
$(\chi,0)(\xi,b)^n=(\xi,b)^n((-nb)\cdot\chi,0)$.

\item $\big((\xi,b)(\chi,0)\big)^n=\big(\prod_{i=1}^n(ib)\cdot\chi,0\big)(\xi,b)^n$.
\end{enumerate}
\end{lemma}
\begin{proof}
The proofs of all three statements are straightforward; see App.~\ref{pf useful formula} for details.
\end{proof}

\section{Properties}\label{sec: group properties}

We are interested in quantum solutions to LCS in $K^{\otimes n}_Q(p)$, \rev{the group generated by Pauli and diagonal Clifford operators}; to this end, we first analyse its abelian subgroups in the next section, before defining a map from $K^{\otimes n}_Q(p)$ to the Heisenberg-Weyl group $H^{\otimes n}(\zz_p)$ that preserves $p$-torsion abelian subgroups in Sec.~\ref{sec: homomorphism of abelian subgroups of order p}. This map will allow us to prove our main result, Thm.~\ref{thm: main result - Clifford hierarchy} in Sec.~\ref{sec: from noncommutative to commutative LCS}. \rev{This section is of a technical nature and may be skipped on first reading. }

\subsection{Abelian subgroups of $K^{\otimes n}_Q(p)$}

From now on,  we fix $p$ to be an odd prime. In this section we will describe ($p$-torsion) abelian subgroups of $K^{\otimes n}_Q(p)$ where $Q=T_{(p^m)}$.  
For simplicity, we will use the following notation:
$$
K=K_Q(p)\;\;\text{ and }\;\; K^{\otimes n}= K^{\otimes n}_Q(p)\; .
$$
We begin with a description of $p$-torsion elements in $K$.

\begin{lemma}\label{lem-torsion-K}
An element $M\in K$ is $p$-torsion, i.e. $M^p=\one$, if either
\begin{itemize}
\item $M=S_\xi$ for some $S_\xi\in T_{(p)}$, or
\item $M=S_\xi X^b$ with $b\neq 0$.
\end{itemize}
\end{lemma}
\begin{proof}
Follows from Lm. \ref{lem-for}(1), and our assumption $S_\xi \in T$, which implies $\det(S_\xi) = \prod_{k=0}^{p-1} k\cdot\xi = 1$.
\end{proof}

Next, we describe commuting pairs of elements in $K$ up to a phase.

\begin{lemma}\label{lm: commutation relations gen HW group up to phase}
Let $M=S_\xi X^b$ and $M'=S_{\xi'}X^{b'}$ be elements in $K$ such that $[M,M']=\omega^c \one$ for some $c\in\zz_p$.         
Then one of the following cases hold:
\begin{enumerate}
\item $b=b'=0$ and $M,M'\in Q$. In particular, $c=0$.
\item $b\neq 0$ and $b'=0$ and  $M'\in\hat{\zz}_p$. In addition, if $c=0$, then $M'\in Z(K)$.
\item  $b,b'\neq 0$ and there exits $a,y\in\zz_p$ and $S_\chi \in\hat{\zz}_p$ such that $\omega^a M'=(M S_\chi)^y$. Here $\chi(q)=\omega^{c'q}$ where $b'c'=-c$ and $yb=b'$. In addition, if $c=0$, then $\omega^a M'=M^y$. 
\end{enumerate}
\end{lemma}

\begin{proof} We will use the alternative notation $(\xi,b)=S
_\xi X^b$. 
First, we calculate the commutator:
\begin{equation}{\label{eqq}}
\begin{aligned}
[M,M']&=(\xi,b)(\xi',b')(\xi,b)^{-1}(\xi',b')^{-1}\\
&=\left(\xi b\cdot\xi',b+b'\right) \,\, \left((-b)\cdot\xi^{-1},-b\right)\left((-b')\cdot\xi'^{-1},-b'\right)\\
&=\left(\xi b\cdot\xi',b+b'\right) \left( (-b)\cdot\xi^{-1}(-b-b')\cdot\xi'^{-1},-b-b'\right)\\
&=\left(\xi b\cdot\xi'(b+b'-b)\cdot\xi^{-1}(b+b'-b-b')\cdot\xi'^{-1},b+b'-b-b'\right)\\
&=\left(\xi b\cdot\xi' b'\cdot\xi^{-1}\xi'^{-1},0\right)\; .
\end{aligned}
\end{equation}
Thus we have $[M,M']=(\hat{\xi},0)$, where $\hat{\xi}(q)=\xi(q)\xi'(q-b)\xi^{-1}(q-b')\xi'^{-1}(q)$ for all $q \in \zz_p$.

Case 1: We assume $b=b'=0$. In this case, we have $M,M'\in Q$. This group is abelian, thus $c=0$. 

Case 2: We assume $b\neq 0$ and $b'=0$. In this case, we have
\begin{align*}
[M,M']=\omega^c \one 	&\iff \xi(q)\xi'(q-b)\xi^{-1}(q)\xi'^{-1}(q)=\omega^c\\
				&\iff \xi'(q-b)=\xi'(q)\omega^c\\
				&\iff \xi'(q)=\xi'(q+b)\omega^c\; .
\end{align*}
Let $m\in\zz_p$ such that $mb=-1$. Thus we have $q+mqb=0$. By the equation above, we have
\[\xi'(q)=\xi'(q+b)\omega^c=\xi'(q+2b)\omega^{2c}=\cdots=\xi'(q+mqb)\omega^{mqc}=\xi'(0)\omega^{mqc}\; .\]
By definition (cf. Eq.~(\ref{eq: maximal torus SU(d)})), we have $\prod_{q=0}^{p-1}\xi'(q)=1$. Thus we obtain
\[1=\prod_{q=0}^{p-1}\xi'(q)=\prod_{q=0}^{p-1}\xi'(0)\omega^{mqc}=(\xi'(0))^p\omega^{mc\sum_{q=0}^{p-1}q}=(\xi'(0))^p\; ,\]
where we used that $\sum_{q=0}^{p-1}q=\frac{p(p-1)}{2}$ is divisible by $2$, hence, $\sum_{q=0}^{p-1}q$ is a multiple of $p$ and thus $\omega^{mc\sum_{q=0}^{p-1}q}=1$.
Finally, we set $\xi'(0)=\omega^a$ for some $a\in\zz_p$ and obtain
\[\xi'(q)=\xi'(0)\omega^{mqc}=\omega^a\omega^{mqc}=\omega^{mcq+a}\; ,\]
therefore $M'=(\xi',b')\in\hat{\zz}_p$. In particular, if $c=0$ we have $\xi'(q)=\omega^a$, and as a result $M'\in Z(K)$. 

Case 3: Assume that $b,b'\neq 0$. We pick $c'\in \zz_p$ and $y\in\zz_p$ such that $b'c'=-c$ and $yb=b'$, and we define $\chi(q)=\omega^{c'q}$. By Lm. \ref{lem-for}(1), we have
\begin{equation}{\label{eq3}}
(MS_\chi)^y=\left((\xi,b)(\chi,0)\right)^y=(\xi b\cdot\chi,b)^y=(\bar{\xi},yb)=(\bar{\xi},b')\; ,
\end{equation}
where $\bar{\xi}=\left(\prod_{i=0}^{y-1}(ib)\cdot \xi\right)\left(\prod_{i=0}^{y-1}(ib+b)\cdot\chi\right)$. Since $[M,M']=\omega^c \one$, by Eq. (\ref{eqq}), we have
\begin{equation}{\label{eq2}}
\xi(q)\xi'(q-b)\xi^{-1}(q-b')\xi'^{-1}(q)=\omega^c\iff \frac{\xi(q)}{\xi(q-b')}=\frac{\xi'(q)}{\xi'(q-b)}\omega^c\; .
\end{equation}
Then we calculate
\begin{align*}
\frac{\bar{\xi}(q)}{\bar{\xi}(q-b)}=\frac{\bar{\xi}}{b\cdot\bar{\xi}}(q)&=\frac{\left(\prod_{i=0}^{y-1}(ib)\cdot\xi\right)\left(\prod_{i=0}^{y-1}(ib+b)\cdot\chi\right)}{\left(\prod_{i=0}^{y-1}(ib+b)\cdot\xi\right)\left(\prod_{i=0}^{y-1}(ib+b+b)\cdot\chi\right)}(q)\\
&=\frac{\xi\left(\prod_{i=1}^{y-1}(ib)\cdot\xi\right)b\cdot\chi\left(\prod_{i=1}^{y-1}(ib+b)\cdot\chi\right)}{\left(\prod_{i=1}^{y}(ib)\cdot\xi\right)\left(\prod_{i=1}^{y}(ib+b)\cdot\chi\right)}(q)\\
&=\left(\frac{\xi}{(yb)\cdot\xi}\right)\frac{b\cdot\chi}{(yb+b)\cdot\chi}(q)\\
&=\left(\frac{\xi(q)}{\xi(q-b')}\right)\frac{\omega^{(q-b)c'}}{\omega^{(q-b'-b)c'}}\,\,\,\,\,\text{(by defintion of $\chi$ and $yb=b'$)}\\
&=\frac{\xi'(q)}{\xi'(q-b)}\omega^c\omega^{b'c'}\,\,\,\,\,\,\text{(by Eq. (\ref{eq2})})\\
&=\frac{\xi'(q)}{\xi'(q-b)}\; .
\end{align*}
By rearranging this equation, we get
\[\frac{\bar{\xi}(q)}{\xi'(q)}=\frac{\bar{\xi}(q-b)}{\xi'(q-b)}\; ,\]
for all $q\in \zz_p$. Thus we conclude that there is a constant $x$ such that  for all $q\in\zz_p$, we have $\bar{\xi}(q)=x\xi'(q)$. By definition, we have $\prod_{q=0}^{p-1}\bar{\xi}(q)=\prod_{q=0}^{p-1}\xi'(q)=1.$ It follows that
\[1=\prod_{q=0}^{p-1}\bar{\xi}(q)=\prod_{q=0}^{p-1}x\xi'(q)=x^p\prod_{q=0}^{p-1}\xi'(q)=x^p\; .\] Thus $x=\omega^a$ for some $a\in \zz_p$ and we have $\bar{\xi}(q)=\omega^a\xi'(q)$ for some $a\in\zz_p$. Combining $\bar{\xi}(q)=\omega^a\xi'(q)$ with Eq. (\ref{eq3}), we obtain
\[(MS_\chi)^y=(\bar{\xi},b')=(\omega^a\xi',b')=(\omega^a,0)(\xi',b')=\omega^aM'\; .\]
In particular, if $c=0$, then $\chi(q)=1$ and $\omega^aM'=M^y$.
\end{proof}

These two results combined gives us the description of $p$-torsion abelian subgroups of $K$.

\begin{corollary}\label{cor: abelian subgroups in K_Q(p)}
Maximal $p$-torsion abelian subgroups of $K$ fall into two classes:
\begin{itemize}
\item the subgroup $T_{(p)}$, and
\item $\Span{\omega\one,S_\xi X}$ where $S_\xi\in Q=T_{(p^m)}$.
\end{itemize} 
Any two distinct maximal $p$-torsion abelian subgroup intersect at the center $Z(K)=\Span{\omega \one}$.
\end{corollary}
\begin{proof}
    This follows immediately from Lm. \ref{lm: commutation relations gen HW group up to phase} with $c = 0$ and Lm. \ref{lem-torsion-K}.
\end{proof}

We reformulate Lm. \ref{lm: commutation relations gen HW group up to phase} in a way that will be useful for computing the commutator in terms of a symplectic form. 

\begin{corollary}\label{cor:commutator-K}
Let $M=S_\xi X^b$ and $M'=S_{\xi'}X^{b'}$ be elements in $K$ such that $[M,M']=\omega^c \one$ for some $c\in\zz_p$.         
Then one of the following cases hold:
\begin{enumerate}[(1)]
\item $b=b'=0$, in which case $c=0$.
\item $b\neq 0$, in which case $M'=S_{\chi'}M^{b'/b}$ where $\chi'(q)=\omega^{\alpha'_0+\alpha'_1q}$ for some $\alpha_0',\alpha_1'\in \zz_p$ and $c=-b\alpha'_1$.
\item $b'\neq 0$, in which case $M=S_{\chi}M^{b/b'}$ where $\chi(q)=\omega^{\alpha_0+\alpha_1q}$ for some $\alpha_0,\alpha_1\in \zz_p$ and  $c=b'\alpha_1$. 
\end{enumerate}
\end{corollary}
\begin{proof}
We will explain the case where at least one of $b$ or $b'$ is nonzero. Then we are either in case (2) or (3) of Lm. \ref{lm: commutation relations gen HW group up to phase}. Assume we are in case (3). We can write 
$
M' = S_{\chi'} M^{b'/b}
$
for some $S_{\chi'}\in\hat{\zz}_p$ using part (2) and (3) of Lm. \ref{lm: commutation relations gen HW group up to phase}. Using Eq. (\ref{eqq}) we compute $c=-b\alpha_1'$. Note that taking $b'=0$ we see that Lm. \ref{lm: commutation relations gen HW group up to phase} part (2) is also covered by this case. Case (3) can be dealt with in a similar way.
\end{proof}

Next we move on to analyzing $K^{\otimes n}$.
Using Thm. \ref{thm-poly} in the Appendix we can express $\xi$ as follows:
\begin{equation}\label{eq:xi}
\xi(q) = \exp\left( \sum_{j=1}^m \frac{2\pi i}{p^j} f_j(q) \right)\; ,
\end{equation}
where each $f_j$ is a polynomial of the form
\begin{equation}\label{eq: polynomial}
	f_j(q) = \sum_{a=0}^{p-1} \vartheta_{j,a} q^a,\;\;\; \vartheta_{j,a}\in \zz_p\; .
\end{equation}
We   need a generalized version of Lm. \ref{lem-torsion-K} to describe $p$-torsion elements in $K^{\otimes n}$.

\begin{lemma}\label{lem:torsion-Kn}
An element $M=M_1\otimes M_2\otimes \cdots\otimes M_n$ in $K^{\otimes n}$ is $p$-torsion if there exists a subset $I\subset \set{1,2,\cdots,n}$ such that the following holds:
\begin{itemize}
\item $M_i\notin Q$, that is, $M_i=S_{\xi_i} X^{b_i}$ with $b_i\neq 0$ for all $i\notin I$.
\item For all $i\in I$ we have $M_i\in Q$  such that $\vartheta_{j,a}^{(i)}=0$ for all $j\geq 2,a\neq 0$ and $\sum_{i\in I} \vartheta_{2,0}^{(i)}=0\mod p$.
\end{itemize} 
\end{lemma}
\begin{proof}
It suffices to show that $M_i\in Q$ satisfies $M_i^p = \omega^{c_i}\one$ for some $c_i \in \zz_p$ whenever $\vartheta_{j,a}^{(i)}=0$ for all $j\geq 2,a\neq 0$. This follows from Eq. (\ref{eq:xi}).
\end{proof}

To track the commutator in the case of $K^{\otimes n}$ we will need the symplectic form borrowed from the Heisenberg-Weyl group $H^{\otimes n}(p)$. For $v=(v_Z,v_X)$ and $v'=(v'_Z,v'_X)$ in $\zz_p^n\times \zz_p^n$ the symplectic form is defined by
$$
[v,v'] = v_Z\cdot v'_X - v'_Z\cdot v_X \mod p\; ,
$$
where $\cdot$ denotes the inner product on vectors in $\zz^n_p$. Using this symplectic form we can describe the commutator of two elements that commute up to a phase.

\begin{corollary}\label{cor:commutator-Kn}
Let $M=M_1\otimes M_2\otimes \cdots \otimes M_n$ and $M'=M'_1\otimes M'_2\otimes \cdots \otimes M'_n$ be elements in $K^{\otimes n}$ such that $[M,M']=\omega^c\one$. For every $1\leq i\leq n$, write $M_i=(\xi_i,b_i)$, let $\bar b=(b_1,b_2,\cdots,b_n)$ and let $\bar f=(\vartheta_{1,1}^{(1)},\vartheta_{1,1}^{(2)},\cdots, \vartheta_{1,1}^{(n)})$, where $\vartheta_{1,1}^{(i)}$ is the coefficient in Eq.~(\ref{eq: polynomial}) that appears in the expansion of $\xi_i$ in Eq.~(\ref{eq:xi}); similarly for $\bar f'$ and $\bar b'$.
Then we have
$$
[M,M'] = \omega^{[v,v']}\one\; ,
$$ 
where  $v=(\bar f,\bar b)$ and  $v'=(\bar f',\bar b')$.
\end{corollary}
\begin{proof}
Since $[M,M']=\omega^c\one$ we have $[M_i,M_i'] = \omega^{c_i}\one$ such that $\sum_{i=1}^n c_i =c \mod p$.
Cor. \ref{cor:commutator-K} implies that $c_i$ is either zero or given by the formula in parts (2) and (3) of this result.
Suppose we are in case (2). Then $M'_i= S_{\chi'} M_i^{b'/b}$ where $\chi' = \alpha_0'+\alpha_1' q$. Using Lm. \ref{lem-for} (1) we can compute the coefficient $\vartheta_{1,1}'^{(i)}$ in $\xi'$.  It turns out that 
$
\vartheta_{1,1}'^{(i)} = \alpha_1' + \frac{b'}{b} \vartheta_{1,1}^{(i)},
$ 
which gives
$$
[M_i,M'_i]=\omega^{-b\alpha_1'}\one = \omega^{-b(\vartheta_{1,1}'^{(i)}-\frac{b'}{b} \vartheta_{1,1}^{(i)})}\one=\omega^{\vartheta_{1,1}^{(i)}b'-\vartheta_{1,1}'^{(i)}b}\one\; .
$$
Case (3) can be dealt with similarly. Combining $\omega^{c_i}$'s we obtain the formula for the symplectic form.
\end{proof}
 
For each $M=M_1\otimes \cdots\otimes M_n$ we will associate the vector $v=(\bar f,\bar b)$ defined in Cor. \ref{cor:commutator-Kn}. A subspace $W$ in $\zz_d^n\times \zz_d^n$ is called isotropic if $[w,w']=0$ for all $w,w'\in W$. Using this definition we observe that
a subgroup $\Span{M^{(1)},M^{(2)}, \cdots,M^{(n)}}$ of $K^{\otimes n}$ is  abelian if and only if  the subspace $\Span{v^{(1)},v^{(2)},\cdots, v^{(n)} }$ of $\zz_d^{2n}$, where $v^{(i)}=(\bar f^{(i)},b^{(i)})$, is isotropic.

\subsection{Projecting $p$-torsion abelian subgroups into the Heisenberg-Weyl group}
\label{sec: homomorphism of abelian subgroups of order p}

Let $K^{\otimes n}_{(p)}$ denote the subset of $p$-torsion elements in $K^{\otimes n}$.
In this section, we will construct a map
$$
\phi: K^{\otimes n}_{(p)} \to H^{\otimes n}(\Z_p)\; ,
$$
that restricts to a group homomorphism on every $p$-torsion abelian subgroup. 
Cor. \ref{cor:commutator-Kn} hints at a ``linearization map" that could potentially serve for this purpose. The map we need is closely related to linearization, yet needs to be slightly adjusted when acting on $p$-torsion elements. We will revisit this point in Remark.~\ref{rm: almost linearisation} below, after the definition of the map.

\begin{definition}{\label{def phi map}}
For $\xi$ given as in Eq. (\ref{eq:xi}) we introduce two maps:
$$
R(\xi) = \omega^{\vartheta_{1,0}+\vartheta_{2,0} + \vartheta_{1,1}q}\;\;\text{ and }\;\; P(\xi) = \omega^{\vartheta_{1,0} + \vartheta_{1,1}q}\; .
$$
Using this notation we define a map 
$$
\phi: K_{(p)}^{\otimes n} \to H^{\otimes n}(\Z_p)
$$
as follows:
$$
\phi(M_1\otimes M_2 \otimes \cdots \otimes M_n) = \phi_1(M_1)\otimes \phi_1(M_2)\otimes \cdots \otimes \phi_1(M_n)\; ,
$$
where $\phi_1(M)$ for $M=(\xi,b)\in K$ is defined by
$$
\phi_1(M) = 
\left\lbrace
\begin{array}{ll}
(R(\xi),0)  &   \text{if $b=0$}\\
(P(\xi),1)  &   \text{if $b=1$}
\end{array}
\right.\; .
$$
and $\phi_1(M) = \phi_1(M^{b_i^{-1}})^{b_i}$ if $1<b<p-1$.
\end{definition}

\begin{remk}\label{rm: almost linearisation}
Note the additional factor $\omega^{\vartheta_{2,0}}$ in the definition of $\phi_1(M)$ for $M\in Q$.  Recall that for $M^p=\one$,
\[M=(\omega^{\sum_{a=0}^{p-1}\vartheta_{1,a}q^a}\sqrt[p]{\omega}^{\vartheta_{2,0}} ,0 )\; ,\]
by Lm. \ref{lem:torsion-Kn}. It is easy to see that defining $\phi_1$ as a linear projection in both cases, that is for $b=0,1$, would not restrict to a group homomorphism as desired.

To see this, let $p=3$ and consider $S_{\xi_1}\otimes S_{\xi_2}$ where $\xi_1(q)= \sqrt[3]{\omega}$ and $\xi_2(q)=\sqrt[3]{\omega}^2$:
$$
\begin{aligned}
\phi(M_1\otimes M_2)\phi(M_1\otimes M_2) &= (\phi_1(M_1)\otimes \phi_1(M_2)) (\phi_1(M_1)\otimes \phi_1(M_2)) \\
& = \one \otimes \one\; ,
\end{aligned}
$$
which is not the same as
$$
\begin{aligned}
\phi((M_1\otimes M_2)(M_1\otimes M_2)) &= \phi( M_1^2 \otimes M_2^2 ) \\
&= \phi( \sqrt[3]{\omega}^2 \one \otimes \omega \sqrt[3]{\omega} \one ) \\
&= \phi_1(\sqrt[3]{\omega}^2 \one)\otimes \phi_1(\omega \sqrt[3]{\omega} \one) \\
&= \omega (\one \otimes \one)\; .
\end{aligned}
$$
\end{remk}

\begin{lemma}{\label{lem-homoq}}
Let $S_\xi,S_{\xi'}\in Q$ be such that $S_{\xi}^p,S_{\xi'}^p \in Z(K)$. Then we have
$$
\phi(S_\xi S_{\xi'}) = \phi(S_\xi)\phi(S_{\xi'})\; .
$$

\end{lemma}
\begin{proof}
The condition $S_{\xi}^p,S_{\xi'}^p \in Z(K)$ implies that $\vartheta_{j,a}=0$ for $j\geq 2, a\neq 0$; similarly for $\vartheta_{j,a}'$ (according to the proof of Lm. \ref{lem:torsion-Kn}).  In other words, we have
\begin{align*}
\xi(q)&=\omega^{\sum_{a=0}^{p-1}\vartheta_{1,a}q^a}\sqrt[p]{\omega}^{\vartheta_{2,0}}\\
\xi'(q)&=\omega^{\sum_{a=0}^{p-1}\vartheta'_{1,a}q^a}\sqrt[p]{\omega}^{\vartheta'_{2,0}}\\
\xi(q)\xi'(q)&=\omega^{\sum_{a=0}^{p-1}(\vartheta_{1,a}+\vartheta_{1,a})q^a}\sqrt[p]{\omega}^{\vartheta_{2,0}+\vartheta'_{2,0}}.\\
\end{align*}
By definition of $\phi$, we have
\[\phi(S_\xi S_{\xi'})(q)=\omega^{\vartheta_{1,0}+\vartheta'_{1,0}+(\vartheta_{1,1}+\vartheta'_{1,1})q}\sqrt[p]{\omega}^{\vartheta_{2,0}+\vartheta'_{2,0}}\; ,\]
which is the same as
\[\phi(S_\xi)(q)\phi(S_{\xi'})(q)=\omega^{\vartheta_{1,0}+\vartheta_{1,1}q}\sqrt[p]{\omega}^{\vartheta_{2,0}}\omega^{\vartheta'_{1,0}+\vartheta'_{1,1}q}\sqrt[p]{\omega}^{\vartheta'_{2,0}}=\omega^{\vartheta_{1,0}+\vartheta'_{1,0}+(\vartheta_{1,1}+\vartheta'_{1,1})q}\sqrt[p]{\omega}^{\vartheta_{2,0}+\vartheta'_{2,0}}\; .\]
Thus $\phi(S_\xi S_{\xi'}) = \phi(S_\xi)\phi(S_{\xi'})$.
\end{proof}

\begin{lemma}{\label{lem-cen}}
Let $M\in K$ such that $M\not\in Q$ and $S_{\chi}\in\hat{\zz}_p$. Then we have
$$
\phi(S_{\chi}M)=\phi(S_{\chi})\phi(M)=S_{\chi}\phi(M)\; .
$$
 
\end{lemma}
\begin{proof}
We denote $M=(\xi,b)$ and $S_{\chi}=(\chi,0)$. 
Define $\bar{\chi}=\prod_{i=0}^{y-1}(ib)\cdot \chi$ and $\bar{\xi}=\prod_{i=0}^{y-1}(ib)\cdot\xi$, where $y\in\Z_p$ such that $yb=1$. 
Consider the element $(\chi,b)$. By Lm. \ref{lem-for}, we have
\[(\chi,b)=((\chi,b)^y)^b=\left(\prod_{i=0}^{y-1}(ib)\cdot\chi,1\right)^b=\left(\prod_{j=0}^{b-1}j\cdot\bar{\chi},b\right)
\implies\chi=\prod_{j=0}^{b-1}j\cdot\bar{\chi}\; .\]
Next, consider the element $(\xi,b)$, we have
\begin{align*}
&(\xi,b)=((\xi,b)^y)^b=\left(\prod_{i=0}^{y-1}(ib)\cdot\xi,1\right)^b=(\bar{\xi},1)^b\\
\implies & \phi(\xi,b)=(P(\bar{\xi}),1)^b=\left(\prod_{j=0}^{b-1}j\cdot P(\bar{\xi}),b\right)\; .
\end{align*}
For the element $S_{\chi}M$, we have
\[(\chi,0)(\xi,b)=((\chi\xi,b)^y)^b=\left(\left(\prod_{i=0}^{y-1}(ib)\cdot\chi\right)\left(\prod_{i=0}^{y-1}(ib)\cdot\xi\right),1\right)^b=(\bar{\chi}\bar{\xi},1)^b.\]
Notice that since $S_{\bar{\chi}}\in\hat{\zz}_p$, we have $P(\bar{\chi}\bar{\xi})=\bar{\chi}P(\bar{\xi})$. Therefore
\begin{align*}
	\phi((\chi,0)(\xi,b))
	=(P(\bar{\chi}\bar{\xi}),1)^b
	&=(\bar{\chi}P(\bar{\xi}),1)^b \\
	&=\left(\left(\prod_{j=0}^{b-1}j\cdot\bar{\chi}\right)\left(\prod_{j=0}^{b-1}j\cdot P(\bar{\xi})\right),b\right)
	=(\chi,0)\phi(\xi,b)
	=\phi(\chi,0)\phi(\xi,b)\; .\qedhere
\end{align*}
\end{proof}

\begin{lemma}{\label{lem}}
Let $M=(\xi,b)$ and $M'=(\xi',b')$ be elements of $K$ such that $[M,M']=\omega^c \one$ for some $c\in\Z_p$ and $b,b'\neq 0$. Then we have 
$$\phi(MM')=\phi(M)\phi(M')\; .$$
\end{lemma}
\begin{proof}
We have $M'=(\xi',b')=(\xi_0',1)^{b'}$ where $\xi'_0=\prod_{i=0}^{z-1}(ib')\cdot\xi'$ and $z\in\Z_p$ such that $zb'=1$. Thus we have 
\[\phi(M')=(P(\xi_0'),1)^{b'}.\]
By Lm. \ref{lm: commutation relations gen HW group up to phase}(3), there exists $a\in\Z_p$ and $y\in \Z_p$ where $yb'=b$ and $S_\chi=(\chi,0)\in\hat{\zz}_p$ such that $M=S_{\omega^a}(M'S_\chi)^y$. Thus we have 
\begin{align*}
M=\omega^a((\xi',b')(\chi,0))^y&=\omega^a\left(\prod_{j=1}^y(jb')\cdot\chi,0\right)(\xi',b')^y\,\,\,\,\,\text{(apply Lm. \ref{lem-for}(3)})\\
&=\omega^a\left(\prod_{j=1}^y(jb')\cdot\chi,0\right)((\xi'_0,1)^{b'})^y\\
&=\omega^a\left(\prod_{j=1}^y(jb')\cdot\chi,0\right)(\xi'_0,1)^b\; .
\end{align*}
By Lm. \ref{lem-cen},  we have 
\[\phi(M)=\omega^a\left(\prod_{j=1}^y(jb')\cdot\chi,0\right)\phi(({\xi'}_0,1)^b)=\omega^a\left(\prod_{j=1}^y(jb')\cdot\chi,0\right)(P(\xi'_0),1)^b\; .\]
Next, consider the product of $M$ and $M'$. We have
\[MM'=\omega^a\left(\prod_{j=1}^y(jb')\cdot\chi,0\right)(\xi'_0,1)^b(\xi_0',1)^{b'}
=\omega^a\left(\prod_{j=1}^y(jb')\cdot\chi,0\right)(\xi'_0,1)^{b+b'}.\]
We consider two cases. First, we assume $b+b'$ is a multiple of $p$. In this case, we have $(\xi'_0,1)^{b+b'}=(P(\xi'_0),1)^{b+b'}=1$. Thus $MM'\in T_{(p)}$. Hence
\[\phi(MM')=\phi\left(\omega^a\left(\prod_{j=1}^y(jb')\cdot\chi,0\right)\right)=\omega^a\left(\prod_{j=1}^y(jb')\cdot\chi,0\right)\; .\]
Therefore we have
\[\phi(M)\phi(M')=\omega^a\left(\prod_{j=1}^y(jb')\cdot\chi,0\right)(P(\xi'_0),1)^{b+b'}=\omega^a\left(\prod_{j=1}^y(jb')\cdot\chi,0\right)=\phi(MM')\; .\]
Second, we assume that $b+b'$ is not a multiple of $p$. By Lm.~\ref{lem-cen}, we then have
\begin{align*}
\phi(MM')=\omega^a\left(\prod_{j=1}^y(jb')\cdot\chi,0\right)\phi((\xi'_0,1)^{b+b'})&=\omega^a\left(\prod_{j=1}^y(jb')\cdot\chi,0\right)(P({\xi'}_0),1)^b(P({\xi'}_0),1)^{b'}\\
&=\phi(M)\phi(M')\; .\qedhere
\end{align*}
\end{proof}

\begin{theorem}\label{thm: quasi-local homomorphism in abelian subgroups of order p}
The map $\phi: K^{\otimes n}_{(p)} \to H^{\otimes n}(\Z_p)$ restricts to a group homomorphism on $p$-torsion abelian subgroups. That is, 
$$
\phi(MM') = \phi(M)\phi(M')\; ,
$$
for all $M,M'\in K^{\otimes n }_{(p)}$ such that $[M,M']=\one$.
\end{theorem}
\begin{proof}
Let $M=M_1\otimes \cdots\otimes M_n$ and $M'=M_1'\otimes\cdots\otimes M'_n$ be elements of $K^{\otimes n}_{(p)}$. Since $M^p=(M')^p=\one$ and $[M,M']=\one$ both Lm. \ref{lem:torsion-Kn} and Cor. \ref{cor:commutator-Kn} apply. In particular, we can reduce to the case $n=1$ according to Def.~\ref{def phi map}.
Consequently, there are four cases to consider:
\begin{itemize}
\item Case 1: $M_i,M'_i\in Q$. Follows from Lm. \ref{lem-homoq}.
\item Case 2: $M_i\in T_{(p)}$ and $M_i\not\in Q$. Follows from Lm. \ref{lem-cen}.
\item Case 3: $M_i\not\in Q$ and $M_i'\in T_{(p)}$. Follows from 
$$
\begin{aligned}
\phi(MM') & = \phi(\omega^\alpha M'M)\\
&= \omega^\alpha \phi(M')\phi(M)\\
&=\omega^\alpha \omega^{-\alpha} \phi(M)\phi(M')\\
&= \phi(M)\phi(M')\; ,
\end{aligned}
$$
where in the first line $\omega^\alpha=[M,M']$, in the second line  we used Lm. \ref{lem-cen} and in the third line we used $[M,M']=[\phi(M),\phi(M')]$ which is a consequence of Cor. \ref{cor:commutator-Kn}.
\item Case 4: $M_i$ and $M_i'$ are not in $Q$. Follows from  Lm. \ref{lem}.\qedhere
\end{itemize}
\end{proof}

\section{Solutions to LCS in $K_Q^{\otimes n}(p)$ are classical}\label{sec: from noncommutative to commutative LCS}

For convenience, we recall the definition of a solution group associated to a LCS from Def. 1 in \cite{QassimWallman2020}:

\begin{definition}\label{def: solution group}
	For a LCS `$Ax=b \mod d$' with $A \in M_{m, n}(\zz_d)$ and $b \in \zz^m_d$, the solution group $\Gamma(A,b)$ is the finitely presented group generated by the symbols ${J, g_1, \cdots, g_n}$ and the relations
	\begin{itemize}
		\item[(a)] $g^d_i = e$, $J^d = e$   (`$d$-torsion'),
		\item[(b)] $Jg_iJ^{-1}g_i^{-1}=e$ for all $i\in\{1,\cdots,n\}$ and $g_jg_kg_j^{-1}g_k^{-1} = e$ whenever there exists a row i such that $L_{ij} \neq 0$ and $L_{ik} \neq 0$,   (`commutativity')
		\item[(c)] $\prod_{j=1}^N g^{A_{ij}}_j = J^{b_i}$ for all i $\{1, \cdots ,m\}$   (`constraint satisfaction').
	\end{itemize}
\end{definition}

We are interested in solutions to LCS in $K^{\otimes n}_Q(p)$ for $p$ odd prime, \rev{the group generated by Pauli and diagonal Clifford operators} (see Def. \ref{def: local MBQC-group}). Surprisingly, we find that LCS admitting a solution $\Gamma \ra K^{\otimes n}_Q(p)$ are classical. To prove this, we will reduce a solution $\Gamma \ra K^{\otimes n}_Q(p)$ to a solution $\Gamma \ra \mH^{\otimes n}(\zz_p)$ using the map $\phi$ from the previous section.

\begin{theorem}\label{thm: main result - Clifford hierarchy}
    Let $\Gamma$ be the solution group of a LCS over $\zz_p$ where $p$ is an odd prime. Then the LCS admits a quantum solution $\eta: \Gamma \ra K^{\otimes n}_Q(p)$ with $T_{(p)}\subset Q\subset T_{(p^m)}$ if and only if it admits a classical solution $\eta: \Gamma \ra \zz_p$.
\end{theorem}

\begin{proof}
    A classical solution $x \in \zz^n_p$ becomes a quantum solution under the identification $M_k = \omega^{x_k}\mathbbm{1} \in Z(K^{\otimes n}_Q(p))$ for all $k\in\{1,\cdots,n\}$. For the converse, assume that $\eta: \Gamma \ra K^{\otimes n}_Q(p)$ is a quantum solution of the LCS. We show that $\phi \circ \eta: \Gamma \ra \mH^{\otimes n}(\zz_p)$ defines a solution of the LCS, i.e., that the map $\phi: K^{\otimes n}_Q(p)_{(p)} \ra \mH^{\otimes n}(\zz_p)$  preserves the constraints in $\Gamma$. This follows directly form Thm. \ref{thm: quasi-local homomorphism in abelian subgroups of order p} more precisely,
    \begin{itemize}
        \item[(i)] $d$-torsion: clearly, every ($n$-qudit Pauli) operator $P \in \mH^{\otimes n}(\zz_p)$ has order $p$ (for $p$ odd prime).
        
        In fact, $\phi$ preserves $d$-torsion  by Thm. \ref{thm: quasi-local homomorphism in abelian subgroups of order p}: $\phi(M)^p = \phi(M^p) = \phi(\one) = \one$ for all $M\in K^{\otimes n}_Q(p)_{(p)}$.
        
        \item[(ii)] commutativity: $\phi$ preserves commutativity by Thm. \ref{thm: quasi-local homomorphism in abelian subgroups of order p}.
        
        \item[(iii)] constraint satisfaction: let $\{M_j\}_{j \in J}$, $M_j \in K^{\otimes n}_Q(p)$ be a set of pairwise commuting operators such that $M_j^p=\one$ for all $j \in J$ and $\omega^{b_J}\one = \prod_{j \in J} M_j$. It follows that $\{M_j\}_{j \in J}$ generates an abelian subgroup of order $p$, hence, $\phi$ preserves constraints by Thm. \ref{thm: quasi-local homomorphism in abelian subgroups of order p}:
\[\omega^{b_J}\one=\phi(\omega^{b_j}\one)=\phi(\prod_{j\in J} M_j)=\prod_{j\in J}\phi(M_j)\; .\]
    \end{itemize}
    Consequently, $\phi \circ \eta: \Gamma \ra \mH^{\otimes n}(\zz_p)$ also defines a solution of the LCS.
    
    Finally, the existence of a classical solution of the LCS follows from Thm. 2 in \cite{QassimWallman2020}, which defines a map $v: \mH(\zz_p) \ra \zz_p$---a homomorphism in abelian subgroups. Similar to above, it follows that $v \circ \phi \circ \eta: \Gamma \ra \zz_p$ yields a classical solution of the LCS described by $\Gamma$.
\end{proof}

More generally, the argument applies to any map $\phi:G\to H$ that restricts to a group homomorphism on $p$-torsion abelian subgroups. 

\begin{corollary}\label{cor: reduced solutions from quasi-local homos}
    Let $\Gamma$ be the solution group of a LCS over $\zz_p$ for $p$ odd prime, let $\eta: \Gamma \ra G$ be a solution and $\phi: G \ra H$ be a map that restricts to a group homomorphism on $p$-torsion abelian subgroups of $G$.
    Then $\phi \circ \eta: \Gamma \ra H$ is a solution to the LCS.
\end{corollary}

\begin{proof}
 Constraints of LCS hold in abelian subgroups of $\Gamma$, hence, are preserved by $\phi$ (Thm. \ref{thm: main result - Clifford hierarchy}).
\end{proof}

\begin{remk}\label{rm: diagonal Clifford hierarchy}
Thm.~\ref{thm: main result - Clifford hierarchy} is a generalisation of Thm.~2 in \cite{QassimWallman2020}. The latter proves that any solution in the Heisenberg-Weyl group $\mc{P}^{\otimes n} \cong \mH^{\otimes n}(\zz_p) \cong K^{\otimes n}_{\hat{\zz}_p}(p) = (\hat{\zz}_p\rtimes \langle X \rangle)^{\otimes n}$ for $p$ odd prime is classical.  Diagonal elements $S = \mathrm{diag}(\xi(0),\cdots,\xi(d-1)) \in \mH(\zz_p)$ are restricted to the Pontryagin dual $\hat{\zz}_p\subset T_{(p)}$; equivalently they lie in the first level of the special diagonal Clifford hierarchy $\mc{SD}_1(d)$ (cf. \cite{CuiGottesmanKrishna2017}). Thm.~\ref{thm: main result - Clifford hierarchy} generalises these to higher levels in the special diagonal Clifford hierarchy $\mc{SD}_k(d)$. In particular, we note that Thm.~\ref{thm-poly} (see App.  \ref{sec: app K_Q(p)}) closely resembles the classification of the diagonal Clifford hierarchy of Thm.~2 in \cite{CuiGottesmanKrishna2017}, as well as the resource analysis of contextuality in $ld$-MBQC of Thm.~5 in \cite{FrembsRobertsCampbellBartlett2023}. Finally, we will discuss the even prime case in App. \ref{even prime}.
\end{remk}

\section{Discussion}\label{sec: discussion}

The Mermin-Peres square in Fig.~\ref{fig: MP square} defines a linear constraint system (LCS) over the Boolean ring $\zz_2$, which admits a quantum solution in terms of two-qubit Pauli operators. To date, no generalisation of the Mermin-Peres square to qudit systems, more generally no finite-dimensional quantum solution to a LCS over $\zz_d$ for $d$ odd, has been found.\footnote{Note that possibly infinite-dimensionsal solutions to LCS over $\zz_d$ for arbitrary $d$ exist by \cite{ZhangSlofstra2020}.} It has been shown that no such examples can be constructed within the Heisenberg-Weyl group, equivalently the qudit Pauli group \cite{QassimWallman2020}.  Our main theorem, Thm.~\ref{thm: main result - Clifford hierarchy}, generalises this result beyond Pauli operators to the diagonal Clifford hierarchy, lending support to the conjecture that no quantum solutions exist in all of $SU(d)$.

\rev{The connection with the Clifford hierarchy} further suggests an avenue either towards generalisations of our result or towards examples of contextual LCS in odd prime dimension. Recall that $K_Q(p)$ contains the diagonal Clifford hierarchy (up to phase). However, $K_Q(p)$ does not contain general Clifford operators.  Still, for $p=3$ every single-qutrit Clifford operator $M$ is semi-Clifford, i.e., $M = C_1DC_2$ with $C_1,C_2 \in \mc{C}_2(3)$ and $D$ diagonal \cite{GrossNest2007,ZengChenChuang2008}. Moreover, this result is conjectured to be true for all prime dimensions \cite{deSilva2021}.
\rev{This hints at a generalisation of our construction to tensor products of groups generated from (semi-)Clifford operators, and, to analyse the existence of quantum solutions to LCS in groups generated from gate sets in quantum computation more generally.}\\

\paragraph*{Acknowledgements.} This work was conducted during a visit of the first author at Bilkent University, Ankara. MF was
supported through grant number FQXi-RFP-1807 from the Foundational Questions Institute and Fetzer Franklin Fund, a donor advised fund of Silicon Valley Community Foundation, and ARC Future Fellowship FT180100317. CO and HC are supported by the Air Force Office of Scientific Research under award number FA9550-21-1-0002.

\appendix

\section{Properties of $K_Q(p)$}\label{sec: app K_Q(p)}

\begin{proof}[Proof of Lm. \ref{lem-for}]
{\label{pf useful formula}}
($1$) We do induction on $n$. The statement is trivial for $n=1$. Assume the statement is true when $n=k$ and consider the case $n=k+1$. We have
\begin{align*}
(\xi,b)^{k+1}=(\xi,b)^k(\xi,b)=\left(\prod_{i=0}^{k-1}(ib)\cdot\xi,kb\right)(\xi,b)&=\left((\prod_{i=0}^{k-1}(ib)\cdot\xi)(kb)\cdot\xi \, ,\, (k+1)b \right)\\
&=\left(\prod_{i=0}^k(ib)\cdot\xi \, , \, (k+1)b\right)\; .
\end{align*}

($2$) We first prove the statement for $n=1$. In this case, we have 
\begin{equation}{\label{eq-comm}}
(\xi,b)(\chi,0)=(\xi b\cdot \chi,b)=(b\cdot\chi,0)(\xi,b)\; .
\end{equation}
Now, we consider the general case. By part $(i)$, we have
\[(\xi,b)^n(\chi,0)=\left(\prod_{i=0}^{n-1}(ib)\cdot\xi \, , \, nb\right)(\chi,0)\stackrel{\text{Eq. (\ref{eq-comm})}}{=}((nb)\cdot\chi,0)\left(\prod_{i=0}^{n-1}(ib)\cdot\xi \, , \, nb\right)=((nb)\cdot\chi,0)(\xi,b)^n\; .\]

($3$) We proceed by induction on $n$. By part ($ii$), the statement holds for $n=1$. Consider $n=k+1$. We have
\begin{align*}
\big((\xi,b)(\chi,0)\big)^{k+1}=\big((\xi,b)(\chi,0)\big)^{k}(\xi,b)(\chi,0)&=\left(\prod_{i=1}^k(ib)\cdot\chi,0\right)(\xi,b)^k(\xi,b)(\chi,0)\\
&=\left(\prod_{i=1}^k(ib)\cdot\chi,0\right)(\xi,b)^{k+1}(\chi,0)\\
&=\left(\prod_{i=1}^k(ib)\cdot\chi,0\right)((k+1)b\cdot\chi,0)(\xi,b)^{k+1}\, \, \, \, \, \, {\text{(by part }(ii))} \\
&=\left(\prod_{i=1}^{k+1}(ib)\cdot\chi,0\right)(\xi,b)^{k+1}\; .
\end{align*}
\end{proof}

Next, we are going to prove that for every $S_\xi\in T_{(p^m)}$, the function $\xi$ can be expressed uniquely as
\begin{equation*}
\xi(q)=\prod_{k=1}^m exp\left[\frac{2\pi i}{p^k}\left(\sum_{a=0}^{p-1}\vartheta_{k,a}q^{a}\right)\right]\; ,
\end{equation*}
where $\vartheta_{i,j}\in \Z_p$. We define the following sets:

\[Q_m^{(p)}=\left\{\begin{pmatrix}x_0 &  & \\ & \ddots & \\ & & x_{p-1}\end{pmatrix}\suchthat \forall i\in {0,\cdots,p-1},\,x_i=exp\left(\frac{2\pi i}{p^m}y_i\right) \text{where $y_i\in\Z_{p^m}$}\right\}\; ,\]

\[
G_m^{(p)}=\left\{S_\xi=\begin{pmatrix}\xi(0) &  & \\ & \ddots & \\ & & \xi(p-1)\end{pmatrix}\suchthat \xi(q)=\prod_{k=1}^m exp\left[\frac{2\pi i}{p^k}\left(\sum_{a=0}^{p-1}\vartheta_{k,a}q^{a}\right)\right]\text{where } \vartheta_{k,a}\in\Z_p\right\}\; .
\]

\begin{remk}
We have the following remarks related to the above definitions:

($i$) Both $Q_m^{(p)}$ and $G_m^{(p)}$ are groups with respect to matrix multiplication.

($ii$) For $i\in\{1,\cdots,p\}$, define a matrix $M_i$ to be as follows
\[(M_i)_{jk}=\begin{cases}
		exp\left(\frac{2\pi i}{p^m}\right) & \text{if $j=k=i$} \\
		1 & \text{if $j=k$ and $j,k\neq i$}\\
		0 & \text{otherwise}
\end{cases}\; .\]
Then $\{M_1,\cdots,M_p\}$ is a generating set of $Q_m^{(p)}$.

($iii$) $T_{(p^m)}$ is a subgroup of $Q_m^{(p)}$. Explicitly, we have $T_{(p^m)}=\{M\in Q_m^{(p)}\suchthat det(M)=1\}$.
\end{remk}

\begin{remk}[Fermat's little theorem]
Let $p$ be a prime. If $a\in\Z$ and $a$ is not divisible by $p$, then we have 
\[a^{p-1}\equiv 1 \mod p\; .\]
\end{remk}

\begin{corollary}{\label{cor-prime}}
Let $p$ be a prime and $q\in\{1,\cdots,p-1\}$. Then there exists $0\neq c\in\N$ such that $1+(p-1)q^{p-1}=cp$.
\end{corollary}
\begin{proof}
Since $p$ is a prime and $q$ is not divisible by $p$, we have
\begin{align*}
&q^{p-1}\equiv 1 \mod p\\
\implies & (p-1)q^{p-1}\equiv p-1 \mod p\\
\implies & 1+(p-1)q^{p-1}\equiv p\equiv 0 \mod p\; .
\end{align*}
Thus there exists $c\in\N$ such that $1+(p-1)q^{p-1}=cp$.
\end{proof}
\begin{theorem}{\label{thm-poly}}
Let $p$ be a prime. For any $m\in\N$, we have $Q_m^{(p)}=G_m^{(p)}$.
In particular, for every element $S_\xi\in T_{(p^m)}$, we can express $\xi$ uniquely as
\begin{equation}{\label{lab-xi}}
\xi(q)=\prod_{k=1}^m exp\left[\frac{2\pi i}{p^k}\left(\sum_{a=0}^{p-1}\vartheta_{k,a}q^{a}\right)\right]\; ,
\end{equation}
where $\vartheta_{k,a}\in\Z_p$ (and $\prod_{q=0}^p \xi(q)=1$). 

\end{theorem}
\begin{proof}
This result can be obtained from Thm. ~2 in \cite{CuiGottesmanKrishna2017}. Here, we provide an independent proof.
It is clear that $G_m^{(p)}\subseteq Q_m^{(p)}$. We are left to show that $Q_m^{(p)}\subseteq G_m^{(p)}$. We claim that there exists $S_{\xi_0}\in G_m^{(p)}$ such that $S_{\xi_0}=diag\left(exp(\frac{2\pi i}{p^{m}}),1,\cdots,1\right)$. Then for $b\in\{1,\cdots,p-1\}$, we define $S_{\xi_b}\in G_m^{(p)}$ by $\xi_b(q)=\xi_0(q-b)$. Thus we have $S_{\xi_b}=diag(1,\cdots,1,exp(\frac{2\pi i}{p^m}),1,\cdots,1)$ where $exp(\frac{2\pi i}{p^m})$ appear in the $b+1^{th}$ position. Observe that any $M\in Q_m^{(p)}$ can be expressed as
\[M=S_{\xi_0}^{k_0}S_{\xi_1}^{k_1}\cdots S_{\xi_{m-1}}^{k_{m-1}}\; ,\]
where $k_0,\cdots,k_{m-1}\in \Z_{p^m}$. Therefore we can define $S_{\bar{\xi}}\in G_m^{(p)}$ where $\bar{\xi}=\prod_{b=0}^{m-1}\xi_b^{k_b}$ such that $S_{\bar{\xi}}=M$. Thus we can conclude that for any $M\in Q_m^{(p)}$, there exists $S_\xi\in G_m^{(p)}$ such that $S_\xi=M$, hence, $Q_m^{(p)}\subseteq G_m^{(p)}$.

Next we prove our claim. We proceed by induction on $m$. 
Consider the case $m=1$, and define
\[\xi(q)=exp\left(\frac{2\pi i}{p}(1+(p-1)q^{p-1})\right)\; .\]
Then for $q=0$, we have $\xi(0)=exp(\frac{2\pi i}{p})$. For $q\in\{1,\cdots,p-1\}$, by Cor. \ref{cor-prime}, there exists $c_q\in\N$ such that
\begin{equation}\label{eq: cor 5}
	1+(p-1)q^{p-1}=c_qp\; .
\end{equation}
Thus for $q\neq 0$, we have 
\[\xi(q)=exp\left(\frac{2\pi i}{p}(1+(p-1)q^{p-1})\right)=exp(2\pi i c_q)=1\; .\] 
Assume the statement is true for $m=n$, and consider the case $m=n+1$. We want to find $S_\xi\in G_{n+1}^{(p)}$ such that $S_\xi=diag\left(exp(\frac{2\pi i}{p^{n+1}}),1,\cdots,1\right)$. Let $M\in Q_n^{(p)}$ be given as
\[M=diag\left(1,exp\left(-\frac{2\pi i}{p^n}(c_1)\right),exp\left(-\frac{2\pi i}{p^n}(c_2)\right),\cdots,exp\left(-\frac{2\pi i}{p^n}(c_{p-1})\right)\right)\; .\]
By induction hypothesis, there exists $S_\theta\in G_n^{(p)}$ where
\[\theta(q)=\prod_{k=1}^n exp\left[\frac{2\pi i}{p^k}\left(\sum_{a=0}^{p-1}c_{k,a}q^{a}\right)\right]\; ,\]
such that
\[S_\theta=diag\left(1,exp\left(-\frac{2\pi i}{p^n}(c_1)\right),exp\left(-\frac{2\pi i}{p^n}(c_2)\right),\cdots,exp\left(-\frac{2\pi i}{p^n}(c_{p-1})\right)\right)\; .\]
Next, we define $S_\xi\in G_{n+1}^{(p)}$ where
\[\xi(q)=\theta(q)exp\left(\frac{2\pi i}{p^{n+1}}(1+(p-1)q^{p-1})\right)\; .\]
For $q=0$, we have $\xi(0)=\theta(0)exp(\frac{2\pi i}{p^{n+1}})=exp(\frac{2\pi i}{p^{n+1}})$. For $q\in\{1,\cdots,p-1\}$, we have 
\[\xi(q)=\theta(q)exp\left(\frac{2\pi i}{p^{n+1}}(1+(p-1)q^{p-1})\right)=exp\left(-\frac{2\pi i}{p^n}c_q\right)exp\left(\frac{2\pi i}{p^{n+1}}(c_qp)\right)=1\; ,\]
where we used Eq.~(\ref{eq: cor 5}). Consequently, there exists $S_\xi\in G_{n+1}^{(p)}$ such that $S_\xi=diag(exp(\frac{2\pi i}{p^{n+1}}),1,\cdots,1)$.
\end{proof}

\section{Even vs odd prime}{\label{even prime}}
In this section, we will discuss the group structure of $K^{\otimes n}_Q(2)$ where $Q=T_{(2^m)}$ for some $m\geq 1$. We will show that Thm. \ref{thm: main result - Clifford hierarchy} does not work for $p=2$.

\begin{lemma}{\label{lem dihedral}}
Let $Q=T_{(2^m)}$. Then $K_Q(2)$ is isomorphic to   the dihedral group $D_{2^{m+1}}$  of order $2^{m+1}$.
\end{lemma}
\begin{proof}
Observe that $T_{(2^m)}$ can be expressed as
$$
T_{(2^m)}=\left\{\begin{pmatrix}a & 0\\0 & a^{-1}\end{pmatrix}\suchthat a\in \CC,\;  a^{2^m}=1\right\}\; ,
$$
which is a cyclic group generated by the matrix $M=\begin{pmatrix}e^{2\pi i/2^m} & 0\\0 & e^{-2\pi i/2^m}\end{pmatrix}$. Therefore 
\begin{equation}{\label{c1}}
K_{T_{(2^m)}}(2)=\left\langle M,X\right\rangle\; .
\end{equation}
Notice that $M$ has order $2^m$ and $X$ has order 2. By simple calculation, we also have $XMX=M^{-1}$. By identifying $X$ with reflection and $M$ with rotation in the dihedral group, we have
\begin{equation*}
K_{T_{(2^m)}}(2)=\langle M,X\rangle\cong \langle r,s\suchthat r^{2^m}=s^2=1, srs=r^{-1}\rangle\cong D_{2^{m+1}}\; .\qedhere
\end{equation*}
\end{proof}

\begin{remk}{\label{remk k4}}
Observe that $K_{T_{(4)}}(2)\cong D_8\cong H(\Z_2)=\langle X,Z\rangle$\footnote{The isomorphism is given by mapping $X\in K_{T_{(4)}}(2)$ to $X\in H(\Z_2)$ and mapping $\begin{pmatrix}i&0\\0 & -i\end{pmatrix}\in K_{T_{(4)}}(2)$ to $XZ\in H(\Z_2)$.} and $K_{T_{(2)}}(2)\cong \zz_2^2$. 
\end{remk}

\begin{lemma}{\label{lem-even prime}}
Assume that $Q=T_{(2^m)}$ where $m\geq 2$. Let $M=\begin{pmatrix}a&0\\0&a^{-1}\end{pmatrix}X^b$ and $M'=\begin{pmatrix}a'&0\\0&{a'}^{-1}\end{pmatrix}X^{b'}$ be elements of $K_Q(2)$ such that $[M,M']=\pm \one\in Z(K_Q(2))$. Then one of the following cases hold:
\begin{enumerate}[(1)]
\item $b=b'=0$ and $M,M'\in Q$. In particular, $[M,M']=\one$.
\item  $b=1$, $b'=0$ and  $M'\in T_{(4)}$. In addition, if $[M,M']=\one$, then $M'\in Z(K_Q(p))$. 
\item $b=0$, $b'=1$ and $M\in T_{(4)}$.  In addition, if $[M,M']=\one$, then $M\in Z(K_Q(p))$. 
\item  $b=b'=1$ and $M=S_\chi M'$ where $S_\chi\in T_{(4)}$. In addition, if $[M,M']=\one$, then $M=\pm M'$.
\end{enumerate}
\end{lemma}
\begin{proof}
First, we consider the case where $b=b'=0$. In this case, we have $M,M'\in Q$. Thus $[M,M']=\one$. Next, by Lm. \ref{lem dihedral}, we can identity $K_Q(2)$ with the dihedral group. In other words, we have $M=r^ks^b$ and $M'=r^{k'}s^{b'}$ for some $k,k'\in\N$, where $r$ and $s$ are the rotation and reflection generator of the dihedral group; respectively. Then our assumption $[M,M']=\pm\one$ can be rewritten as $[r^ks^b,r^{k'}s^{b'}]=r^n$ where $r^{2n}=1$. 

First, we consider $b=1$ and $b'=0$. In this case, we have $M=r^ks$ and $M'=r^{k'}$. By calculation, we get 
$$r^{-n}=r^{n}=[M,M']=(r^ks)(r^{k'})(sr^{-k})(r^{-k'})=r^kr^{-k'}r^{-k}r^{-k'}=r^{-2k'}\Rightarrow (r^{k'})^4=1\; .$$
Thus we can conclude that $M'^4=\one$ and therefore $M'\in T_{(4)}$. In addition, if $[M,M']=\one$, we have $1=r^{2k'}$, then $M'^2=\one$. It follows that $M'=\pm \one\in Z(K_Q(2))$. The case $b=0$ and $b'=1$ is similar.

Next, consider $b=b'=1$. In this case, we have $M=r^ks$ and $M'=r^{k'}s$. By calculation, we get 
$$r^n=[M,M']=(r^ks)(r^{k'}s)(sr^{-k})(sr^{-k'})=r^kr^{-k'}r^kr^{-k'}=r^{2(k-k')}\; .$$
We can now conclude that $(r^{k-k'})^4=1$ and $r^{k'}=r^{-n}r^{k-k'}r^k$. Observe that $(r^{-n}r^{k-k'})^4=1$. Thus $M=S_\chi M'$ for some $S_\chi\in T_{(4)}$. In addition, if $[M,M']=\one$, we have $(r^{k-k'})^2=1$ and therefore $M=\pm M'$.
\end{proof}
\begin{corollary}\label{cor:poset-2}
Maximal $2$-torsion abelian subgroups of $K_Q(2)$ fall into two classes:
\begin{itemize}
\item the subgroup $T_{(2)}$, and
\item $\Span{-\one,S_\xi X}$ where $S_\xi\in Q=T_{(2^m)}$.
\end{itemize} 
Any two distinct maximal $2$-torsion abelian subgroups intersect at the center $Z(K_Q(2))=\{\pm\one\}$.
\end{corollary}
\begin{proof}
    This follows immediately from Lm. \ref{lem-even prime} with $[M,M']=\one$.
 \end{proof}

\begin{example}\label{ex: p=2}
Define $M=S_{\xi_1} X$ and $N=S_{\xi_2}X$ where
\begin{align*}
\xi_1(q)&=(-1)^{1+q}i\; ,\\
\xi_2(q)&=(-1)^qi\; .
\end{align*}
In other words, we have $M=\begin{pmatrix}0&-i\\ i&0\end{pmatrix}$ and $N=\begin{pmatrix}0&i\\ -i&0\end{pmatrix}$. Notice that $M^2=N^2=\one$ and $[M,N]=\one$. Suppose that we define $\phi$ for $p=2$ analogous to the odd case, more precisely using the $P$ map in Def. \ref{def phi map}, we have
\begin{align*}
\phi(M)=\begin{pmatrix}-1 & 0\\0 &1\end{pmatrix}X=\begin{pmatrix}0 & -1\\1 &0\end{pmatrix}\; ,\\
\phi(N)=\begin{pmatrix}1 & 0\\0 &-1\end{pmatrix}X=\begin{pmatrix}0 & 1\\-1 &0\end{pmatrix}\; .
\end{align*} 
Thus we have $\phi(M)\phi(N)=\one$. On the other hand, we have $MN=-\one$ and therefore 
\[\phi(MN)=-\one\neq \phi(M)\phi(N)\; .\]
Notice that both $\phi(M)$ and $\phi(N)$ are elements of $H(\Z_2)$ rather than $K_{T_{(2)}}(2)$.
\end{example}

\begin{remk}
Thm. \ref{thm: main result - Clifford hierarchy} does not work for $p=2$.  Note first that unlike the odd prime case (cf. Rmk.~\ref{rm: diagonal Clifford hierarchy}),  for $p=2$ we have $\mc{P}_2 \not\cong H(\zz_2) 
\not\cong K_{T_{(2)}}(2) = T_{(2)} \rtimes \langle X\rangle = \langle -\one,X\rangle$.
Now, as defined in Def. \ref{def phi map}, $\phi$ maps into $H^{\otimes n}(\zz_2)$ rather than $\langle -\one,X\rangle^{\otimes n}$, yet it is easy to see that $\phi: K^{\otimes n}_{T_{(2)}} \ra H^{\otimes n}(\zz_2)$ is not a homomorphism when restricted to $2$-torsion abelian subgroups (see Ex.~\ref{ex: p=2} above).

Moreover,  there is no way to adapt Def.  \ref{def phi map} in such a way that $\phi$ maps $K_Q^{\otimes n}(2)$ into $\langle -\one,X\rangle^{\otimes n}$ for $n>1$. If that were the case, analogous arguments to those in Thm.~\ref{thm: main result - Clifford hierarchy} would imply that $K^{\otimes n}_Q$ is noncontextual, yet $K^{\otimes n}_Q$ is contextual. To see this, we remark that the quantum solution corresponding to the LCS given by Mermin-Peres square already arises for $H(\zz_2)^{\otimes 2}$.\footnote{In fact, it is sufficient to note that $Y \otimes Y = - ZX \otimes ZX$ in Fig. \ref{fig: MP square} (a).} Consequently, $H(\zz_2)^{\otimes n}$ is contextual for all $n > 1$. Finally, using the isomorphism $K_{(4)} \cong H(\zz_2)$ (cf. Rmk. \ref{remk k4}) it also follows that $K_{(2^m)}^{\otimes n}$ is contextual for all $n>1,m>2$. As a final remark note that  Corollary \ref{cor:poset-2} implies that $K_Q(2)$ is noncontextual. 
\end{remk}


\begin{thebibliography}{10}

\bibitem{AaronsonGottesman2004}
Scott Aaronson and Daniel Gottesman.
\newblock Improved simulation of stabilizer circuits.
\newblock \href{https://doi.org/10.1103/PhysRevA.70.052328}{\em Phys. Rev. A, 70:052328}, (Nov 2004).

\bibitem{Arkhipov2012}
Alex Arkhipov.
\newblock Extending and characterizing quantum magic games.
\newblock \href{https://doi.org/10.48550/arXiv.1209.3819}{\em arXiv.1209.3819} (2012).

	
\bibitem{BengtssonBlanchfieldCampbellHoward2014}
Ingemar Bengtsson, Kate Blanchfield, Earl Campbell, and Mark Howard.
\newblock Order 3 symmetry in the clifford hierarchy.
\newblock \href{https://doi.org/10.1088/1751-8113/47/45/455302}{\em J. Phys. A, 47(45):455302}, (2014).

\bibitem{CleveLiuSlofstra2017}
Richard Cleve, Li~Liu, and William Slofstra.
\newblock Perfect commuting-operator strategies for linear system games.
\newblock \href{
https://doi.org/10.1063/1.4973422
}{\em J. Math. Phys. 58(1):012202}, (2017).

\bibitem{CleveMittal2014}
Richard Cleve and Rajat Mittal.
\newblock Characterization of binary constraint system games.
\newblock  \href{https://doi.org/10.1007/978-3-662-43948-7_27}{\em In Javier Esparza, Pierre Fraigniaud, Thore Husfeldt, and Elias
  Koutsoupias, editors, Automata, Languages, and Programming, pages
  320--331, Berlin, Heidelberg, 2014. Springer Berlin Heidelberg}

\bibitem{CuiGottesmanKrishna2017}
S.~X. {Cui}, D.~{Gottesman}, and A.~{Krishna}.
\newblock {Diagonal gates in the Clifford hierarchy}.
\newblock \href{https://doi.org/10.1103/PhysRevA.95.012329}{\em Phys. Rev. Lett., 95(1):012329}, (2017).

\bibitem{deBeaudrap2013}
Niel de~Beaudrap.
\newblock A linearized stabilizer formalism for systems of finite dimension.
\newblock \href{
https://doi.org/10.26421/QIC13.1-2-6}{\em Quantum Info. Comput., 13(1-2):73--115}, (Jan 2013).

\bibitem{deSilva2021}
Nadish de~Silva.
\newblock Efficient quantum gate teleportation in higher dimensions.
\newblock \href{https://doi.org/10.1098/rspa.2020.0865}{\em Proc. Math. Phys. Eng. Sci., 477(2251):20200865}, (2021).

\bibitem{FrembsRobertsBartlett2018}
Markus Frembs, Sam Roberts, and Stephen~D. Bartlett.
\newblock Contextuality as a resource for measurement-based quantum computation
  beyond qubits.
\newblock \href{
https://doi.org/10.1088/1367-2630/aae3ad
}{\em New J. Phys., 20(10):103011}, (Oct 2018).

\bibitem{FrembsRobertsCampbellBartlett2023}
Markus Frembs, Sam Roberts, Earl~T Campbell, and Stephen~D Bartlett.
\newblock Hierarchies of resources for measurement-based quantum computation.
\newblock \href{
https://doi.org/10.1088/1367-2630/acaee2
}{\em New J. Phys., 25(1):013002}, (Jan 2023).

\bibitem{Fritz2012}
Tobias Fritz.
\newblock {Tsirelson’s problem and Kirchberg’s conjecture}.
\newblock \href{
https://doi.org/10.1142/S0129055X12500122}{\em Rev. Math. Phys., 24(05):1250012}, (2012).

\bibitem{Fritz2016}
Tobias Fritz.
\newblock Quantum logic is undecidable.
\newblock \href{
https://doi.org/10.1007/s00153-020-00749-0}{\em  Arch. Math. Logic 60, 329-341}, (2021).

\bibitem{Fujii2015}
Keisuke Fujii.
\newblock Stabilizer Formalism and Its Applications, pages 24--55.
\newblock \href{https://link.springer.com/chapter/10.1007/978-981-287-996-7_2}{\em Springer Singapore, Singapore}, (Dec 2015).

\bibitem{Gottesman1997}
Daniel Gottesman.
\newblock Stabilizer codes and quantum error correction.
\newblock \href{
https://doi.org/10.48550/arXiv.quant-ph/9705052}{\em  arXiv:quant-ph/9705052} (1997).

\bibitem{Gottesman1998}
Daniel Gottesman.
\newblock Theory of fault-tolerant quantum computation.
\newblock \href{
https://doi.org/10.1103/PhysRevA.57.127}{\em Phys. Rev. A, 57:127--137}, (Jan 1998).

\bibitem{GottesmanChuang1999}
Daniel Gottesman and Isaac~L. Chuang.
\newblock Demonstrating the viability of universal quantum computation using
  teleportation and single-qubit operations.
\newblock \href{https://doi.org/10.1038/46503}{\em Nature, 402(6760):390--393}, (Nov 1999).

\bibitem{GrossNest2007}
David Gross and Maarten Nest.
\newblock {The LU-LC conjecture, diagonal local operations and quadratic forms
  over $GF(2)$}.
\newblock \href{https://dl.acm.org/doi/10.5555/2011763.2011766}{\em Quantum Information \& Computation, 8:263--281}, (Mar 2008).

\bibitem{Hall}
Brian~C. Hall.
\newblock {\em Lie Groups, Lie Algebras, and Representations: An Elementary
  Introduction}.
\newblock \href{https://doi.org/10.1007/978-3-319-13467-3}{\em Springer}.

\bibitem{Howard2015}
Mark Howard.
\newblock Maximum nonlocality and minimum uncertainty using magic states.
\newblock \href{https://doi.org/10.1103/PhysRevA.91.042103}{\em Phys. Rev. A, 91:042103}, (Apr 2015).

\bibitem{HowardVala2012}
Mark Howard and Jiri Vala.
\newblock Qudit versions of the qubit $\ensuremath{\pi}/8$ gate.
\newblock \href{https://doi.org/10.1103/PhysRevA.86.022316}{\em Phys. Rev. A, 86:022316}, (Aug 2012).

\bibitem{JiEtAl2020}
Zhengfeng Ji, Anand Natarajan, Thomas Vidick, John Wright, and Henry Yuen.
\newblock $\mathrm{MIP}^*=\mathrm{RE}$.
\newblock \href{https://doi.org/10.1145/3485628}{\em Communications of the ACM, Volume 64, Issue 11 (2021), p131-138}.

\bibitem{JungeEtAl2011}
M.~Junge, M.~Navascues, C.~Palazuelos, D.~Perez-Garcia, {V. B.} Scholz, and {R.
  F.} Werner.
\newblock {Connes' embedding problem and Tsirelson's problem}.
\newblock \href{https://doi.org/10.1063/1.3514538}{\em J. Math. Phys., 52(1)}, (Jan 2011).

\bibitem{Kirchberg1993}
Eberhard Kirchberg.
\newblock {On non-semisplit extensions, tensor products and exactness of group
  $C^*$-algebras}.
\newblock \href{https://doi.org/10.1007/BF01232444}{\em Inventiones mathematicae, 112:449--489}, (1993).

\bibitem{Mermin1990}
N.~David Mermin.
\newblock Simple unified form for the major no-hidden-variables theorems.
\newblock \href{https://doi.org/10.1103/PhysRevLett.65.3373}{\em Phys. Rev. Lett., 65:3373--3376}, (Dec 1990).

\bibitem{Mermin1993}
N.~David Mermin.
\newblock {Hidden variables and the two theorems of John Bell}.
\newblock \href{
https://doi.org/10.1103/RevModPhys.65.803}{\em Rev. Mod. Phys., 65:803--815}, (Jul 1993).

\bibitem{OkayRaussendorf2020}
Cihan Okay and Robert Raussendorf.
\newblock Homotopical approach to quantum contextuality.
\newblock \href{https://doi.org/10.22331/q-2020-01-05-217}{\em Quantum, 4:217}, (Jan 2020).

\bibitem{QassimWallman2020}
Hammam Qassim and Joel~J Wallman.
\newblock Classical vs quantum satisfiability in linear constraint systems
  modulo an integer.
\newblock \href{https://doi.org/10.1088/1751-8121/aba306}{\em J. Phys. A, 53(38):385304}, (Aug 2020).

\bibitem{Raussendorf2013}
Robert {Raussendorf}.
\newblock {Contextuality in measurement-based quantum computation}.
\newblock \href{https://doi.org/10.1103/PhysRevA.88.022322}{\em Phys. Rev. A, 88(2):022322}, (Aug 2013).

\bibitem{Slofstra2019b}
William Slofstra.
\newblock The set of quantum correlations is not closed.
\newblock \href{https://doi.org/10.1017/fmp.2018.3}{\em Forum of Mathematics, Pi, 7:41}, (Jan 2019).

\bibitem{Slofstra2019}
William Slofstra.
\newblock Tsirelson’s problem and an embedding theorem for groups arising
  from non-local games.
\newblock \href{
https://doi.org/10.1090/jams/929}{\em J. Am. Math. Soc., 33(1):1–56}, (Sep 2019).

\bibitem{ZhangSlofstra2020}
William Slofstra and Luming Zhang.
\newblock \textit{private communication}.

\bibitem{Tsirelson2006}
Boris~S. Tsirelson.
\newblock Bell inequalities and operator algebras.
\newblock \href{http://web.archive.org/web/20090414083019/http://www.imaph.tu-bs.de/qi/problems/33.html}{\em This was posted by Boris Tsirelson as problem, 33.}, (2006)

\bibitem{ZengChenChuang2008}
Bei Zeng, Xie Chen, and Isaac~L Chuang.
\newblock {Semi-Clifford operations, structure of $C(k)$ hierarchy, and gate
  complexity for fault-tolerant quantum computation}.
\newblock \href{
https://doi.org/10.1103/PhysRevA.77.042313}{\em Phys. Rev. A, 77(4):042313}, (Apr 2008).

\end{thebibliography}
\end{document}